\documentclass[a4paper,11pt]{article}
\usepackage[utf8]{inputenc}
\usepackage[T1]{fontenc}
\usepackage{amsmath}
\usepackage{epsfig}
\usepackage{epic}
\usepackage{setspace}

\newtheorem{theorem}{Theorem}

\newtheorem{proposition}[theorem]{Proposition}
\newenvironment{proof}[1][Proof]{\noindent \textbf{#1.} }{\  \rule{0.5em}{0.5em}}
\newcommand{\dlim}{\xrightarrow{\cal L}}

\newcommand{\ym}{{\bf y}_{t-}}
\newcommand{\bphi}{{\boldsymbol\phi}}
\newcommand{\bu}{{\bf U}}
\newcommand{\te}{{\tilde{\boldsymbol\varepsilon}}}
\newcommand{\ve}{{\rm vec}}
\newcommand{\bPhi}{{\boldsymbol\Phi}_0}

\begin{document}

\title{Confidence distributions for the  parameters in an autoregressive process}
\author{Rolf Larsson\footnote{Dept of Mathematics, Uppsala University, P.O.Box 480, SE-751 06 Uppsala, Sweden, rolf.larsson@math.uu.se.}}
\maketitle

\begin{abstract}
We suggest how to construct joint confidence distributions for several parameters and apply these ideas to an autoregressive process of general order. The implied non informative prior for the parameters, i.e. the ratio between the confidence density and the likelihood function, is proved to be asymptotically flat in the stationary case. However, in the presence of a unit root, the implied prior needs to be adjusted. The results are illustrated by simulation studies and empirical examples.
\vskip.2cm\noindent
{\it Keywords}: Confidence distributions, Autoregressive processes, Unit root test, Asymptotic normality, Non informative prior.
\end{abstract}

\section{Introduction}

Let us say that we have a set of observations (data) from a statistical model with some unknown parameter, $\theta$ say, that needs to be estimated. Bayesian statistical inference tackles this problem by assuming a prior distribution for $\theta$, and when multiplying by the likelihood function obtained from data and normalizing, the posterior distribution is gained. For a historical account on early Bayesian thinking, see e.g. Hald (2007).

In a seminal paper, Fisher (1930) introduced fiducial inference, which is a methodology that gives the distribution of $\theta$ without having to assume any prior. This was accomplished by inversion of a pivotal quantity. Later on, Fisher (1930) was heavily criticized for his idea, but his work inspired the notion of confidence distributions, cf e.g. Neyman (1941) and Cox (1958). Confidence distributions accomplish what Fisher tried to do in a more frequentist way of thinking, by relating to the concepts of confidence and P values for hypothesis testing. For more modern references, see e.g. Xie and Singh (2013), Schweder and Hjort (2016) (SH in the following) and Pawitan and Lee (2021). In particular, the latter paper discusses the notion of implied prior. The implied prior is obtained as the ratio between the confidence density and the likelihood function, and it can be seen as a non informative prior in the sense of Jeffreys (1931).

Larsson (2024) focused on autoregressive processes of order one. The main result was to show that in the stationary case, asymptotically the implied prior for the AR parameter is flat. In the present paper, we aim to generalize this result to an autoregressive process of general order. By simulation, we will also compare different methods of calculating the confidence region in terms of coverage and area. 

In the rest of the paper, section 2 gives a general presentation of the confidence distribution framework in the one parameter case, and discusses its generalization to several parameters. Section 3 specializes into a theoretical treatment of the autoregressive model. Section 4 provides simulation results, section 5 gives two applications and section 6 concludes. Omitted proofs and extra figures are provided in the Appendix.

\section{Confidence distributions}

At first, assume a one parameter model, i.e. we have a statistical model with a single parameter $\theta$ and a random sample $y_1,...,y_n$. Suppose that for a given $\theta_0$, we want to test $H_0$: $\theta\leq\theta_0$ vs $H_1$: $\theta>\theta_0$. Basing the test on some estimator $\hat\theta_{obs}$, for example the maximum likelihood estimator (MLE), and rejecting for large values, the P value of the test is given by
\begin{equation}
C(\theta_0)=P_{\theta_0}(\hat\theta\geq\hat\theta_{obs}),\label{Pvalue}
\end{equation}
where we calculate the probability under $\theta=\theta_0$. Moreover, $\hat\theta$ is the random variable of which $\hat\theta_{obs}$ is an observation. 

By construction, it is clear that given $\hat\theta_{obs}$, a small $\theta_0$ gives a small $\hat\theta$ and a small $C(\theta_0)$. Conversely, a large $\theta_0$ corresponds to $C(\theta_0)$ being close to one. Hence, it is natural that $C(\theta_0)$ is a monotone increasing function of $\theta_0$ taking values in the unit interval, i.e. it acts as a distribution function with argument $\theta_0$. Another way of putting this is to say that the larger the P value, the larger is the support for the null hypothesis that $\theta\leq\theta_0$. In this sense, $C(\theta_0)$ gives the {\it confidence} for $\theta\leq\theta_0$.
We say that $C(\theta)$ is a {\it confidence distribution function} of the parameter $\theta$. From now on, we will assume that the underlying random variable is continuous, so that $C(\theta)$ is a continuous function of $\theta$. 

Related concepts are the {\it confidence density} of $\theta$, obtained by differentiating $C(\theta)$ with respect to $\theta$, i.e
\begin{equation}
c(\theta)=\frac{d}{d\theta}C(\theta),\label{c}
\end{equation}
and the {\it confidence curve}
\begin{equation}
cc(\theta)=|1-2C(\theta)|,\label{cc}
\end{equation}
that depicts all equally tailed confidence intervals for $\theta$. For a more thorough and formal discussion of these concepts, see SH.

The confidence density $c(\theta)$ may be interpreted as a posterior density for $\theta$ in conjunction with a prior density $g(\theta)$ (the implied prior) that satisfies
\begin{equation}
c(\theta)=L(\theta)g(\theta),\label{impprior}
\end{equation}
where $L(\theta)$ is the likelihood, cf Pawitan and Lee (2021).

Now, say that the model has $m$ unknown parameters of interest, $\theta_1,...,\theta_m$. So far, the literature contains surprisingly little discussion, if any, of this case. For example, SH only discusses this multiparameter problem in the context of a focus parameter, i.e. one (function of the) parameter(s) is of interest, and the others act as nuisance parameters.

Let us try to discuss the multiparameter case intuitively, and for the sake of argument, we start with $m=2$. At first, we generalize the P value equation (\ref{Pvalue}) into
\begin{equation}
C(\theta_{10},\theta_{20})=P_{\theta_{10},\theta_{20}}(\{\hat\theta_1\geq\hat\theta_{1,obs}\}\cap\{\hat\theta_2\geq\hat\theta_{2,obs}\}).\label{Pvalue2}
\end{equation}

We 
have the important intuition that, just as $C(\theta_0)$ in (\ref{Pvalue}) may be interpreted as giving the confidence for the null hypothesis $\theta\leq\theta_0$, $C(\theta_{10},\theta_{20})$ in (\ref{Pvalue2}) gives the confidence for the null hypothesis $\theta_1\leq\theta_{10}$ and $\theta_2\leq\theta_{20}$ in the case $m=2$. This in turn corresponds to the joint confidence distribution function of $(\theta_1,\theta_2)$.  

Note that if the intersection $\cap$ in (\ref{Pvalue2}) is replaced by the union $\cup$, we do not obtain a joint distribution function.

The joint confidence density of $(\theta_1,\theta_2)$ would in this case be the mixed derivative
\begin{equation}
c(\theta_1,\theta_2)=\frac{\partial^2}{\partial\theta_1\partial\theta_2}C(\theta_1,\theta_2).\label{c2}
\end{equation}
Generalizing the univariate confidence curve of (\ref{cc}) is less straightforward, since if we go beyond rectangular shapes, it is hard to establish the meaning of equally tailed confidence regions. One suggestion is to construct the $1-\alpha$ confidence region for $(\theta_1,\theta_2)$ as the region with $c(\theta_1,\theta_2)\geq k$ for a constant $k$ such that the confidence of the region equals $1-\alpha$. Then, the confidence 'curve' is the ensemble of such regions obtained when varying the confidence.

Another way to obtain a confidence 'curve' is to use a test statistic, see section 3.

It is clear that we may generalize the confidence distribution function for $m=2$ to the case of a general $m$ as
\begin{equation}
C(\theta_{10},...,\theta_{m0})=P_{\theta_{10},\theta_{20}}(\{\hat\theta_1\geq\hat\theta_{1,obs}\}\cap...\cap\{\hat\theta_m\geq\hat\theta_{m,obs}\}),\label{Pvaluem}
\end{equation}
with corresponding generalizations of the joint confidence density and the confidence 'curve'.

\section{The autoregressive model}

\subsection{Preliminaries}

Consider an autoregressive process of order $p$,
\begin{equation}
Y_t=\phi_1 Y_{t-1}+...+\phi_p Y_{t-p}+\varepsilon_t,\label{ARp}
\end{equation}
where $t=1,2,...$ and the $\varepsilon_t$ are independent and normally distributed with expectation 0 and unknown variance $\sigma^2$. Say that we have observations $y_1,...,y_n$.  Assume that $y_0=y_{-1}=...=y_{1-p}=0$.

By writing (\ref{ARp}) as a regression model,
\begin{equation}
{\bf y}={\bf X}{\boldsymbol\phi}+{\boldsymbol\varepsilon},\label{reg}
\end{equation}
where 
$${\bf y}=\left(\begin{array}{c}y_1\\y_2\\ \vdots \\y_n\end{array}\right),
\
{\bf X}=\left(\begin{array}{cccc}y_0&y_{-1}&\hdots&y_{1-p}\\
y_1&y_0&\hdots&y_{2-p}\\
\vdots&\vdots&&\vdots\\
y_n&y_{n-1}&\hdots&y_{n+1-p}\end{array}\right),\
{\boldsymbol\phi}=\left(\begin{array}{c}\phi_1\\ \phi_2\\ \vdots \\ \phi_p\end{array}\right),
\
{\boldsymbol\varepsilon}=\left(\begin{array}{c}\varepsilon_1\\\varepsilon_2\\ \vdots \\\varepsilon_n\end{array}\right),
$$
we obtain the least squares estimates (also the maximum likelihood estimates, the MLEs, under the initial value assumption above)
\begin{equation}
\hat{\boldsymbol\phi}=({\bf X}'{\bf X})^{-1}{\bf X}'{\bf y},\label{LSE}
\end{equation}
where $\hat{\boldsymbol\phi}=(\hat\phi_1,...,\hat\phi_p)'$.

Suppose that the true parameters are $\bphi_0=(\phi_{10},...,\phi_{p0})'$.
Assume that $\{y_t\}$ is stationary for a suitable choice of initial distribution. Then, by e.g. Shumway and Stoffer (2017) (SS in the following) p.126, as $n\to\infty$,
$$\sqrt{n}(\hat\bphi-\bphi_0)\dlim N\left({\bf 0},\Omega_0\right),$$
where $\dlim$ denotes convergence in distribution, ${\bf 0}$ is a zero vector and $\Omega_0$ is the asymptotic covariance matrix, which in turn is a function of $\bphi_0$. We may replace $\Omega_0$ by $\hat\Omega$, where we have inserted the observed estimate $\hat\bphi_{obs}$ for $\bphi_0$ in $\Omega_0$, to obtain the further asymptotic result
\begin{equation}
\sqrt{n}(\hat\bphi-\bphi_0)\dlim N\left({\bf 0},\hat\Omega\right).\label{MLas}
\end{equation}
Hence, $\hat\bphi-\bphi_0$ is approximately $N({\bf 0}_p,n^{-1}\hat\Omega)$, where ${\bf 0}_p$ is a $p$ dimensional vector of zeros. In dimension $p=2$ (SS, p. 126),
$$\hat\Omega=\left(\begin{array}{cc}1-\hat\phi_2^2&
-\hat\phi_1(1+\hat\phi_2)\\-\hat\phi_1(1+\hat\phi_2)&1-\hat\phi_2^2
\end{array}\right).$$

Say that we want to test the null hypothesis $H_0$: $\boldsymbol\phi=\boldsymbol\phi_0$ vs the alternative $H_1$: $\neg H_0$. The Wald statistic for this test is
\begin{equation}
Q=(\hat\bphi-\bphi_0)'\hat\Omega^{-1}(\hat\bphi-\bphi_0),\label{Wald}
\end{equation}
and it follows easily from (\ref{MLas}) that as $n\to\infty$, the asymptotic distribution of $Q$ is $\chi^2$ with $p$ degrees of freedom, $\chi^2_p$. All $\boldsymbol\phi_0$ that are not rejected by the test of $H_0$ vs $H_1$ at significance level $\alpha$ constitute a  confidence region for $\boldsymbol\phi$ with confidence level $1-\alpha$. Hence, for large $n$ and parameters well inside of the stationarity region, we obtain a confidence curve with approximate confidence level $1-\alpha$ by including all $\boldsymbol\phi=\boldsymbol\phi_0$ that satisfy $Q\leq\chi^2_p(\alpha)$, where $Q$ is as in (\ref{Wald}) and such that $P\{X\geq \chi^2_p(\alpha)\}=\alpha$, where $X$ is distributed as $\chi^2_p$. For smaller $n$ or when the parameters are close to the border of the stationarity region, a better approximation should be obtained by bootstrapping the percentile.

Since the confidence curve depicts confidence regions for all confidence levels, the Wald statistic may be readily used for its construction. In Figures 1-4, we illustrate such confidence curves for dimension $p=2$, obtained by simulation from (\ref{ARp}) with $\sigma^2=1$ and $n=100$. Figure 1 and 2 depict the case 
$\boldsymbol\phi_0=(0.4,\ 0.2)$, for which we obtained the MLE 
$\hat\bphi_0\approx (0.37,\ 0.25)$ in this simulation. 
In Figure 1 we use the $\chi^2$ approximation and in Figure 2 we employ the bootstrap, with 1000 bootstrap samples. Corresponding to these are Figures 3 and 4, but here $n=50$, 
$\boldsymbol\phi_0=(0.6,\ 0.3)$ with MLE $\hat\bphi_0\approx (0.49,\ 0.41)$.

\begin{figure}[htb]
\begin{center}
\includegraphics[width=90mm]{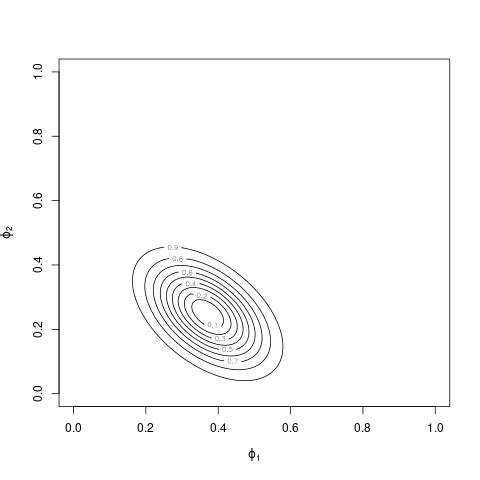}
\end{center}
\caption{Simulated confidence curve obtained by $\chi^2$ approximation of the Wald statistic for the case $n=100$, $\sigma^2=1$, $\boldsymbol\phi_0=(0.4,\ 0.2)$, $\hat\bphi_0\approx (0.37,\ 0.25)$.}
\end{figure}

\begin{figure}[htb]
\begin{center}
\includegraphics[width=90mm]{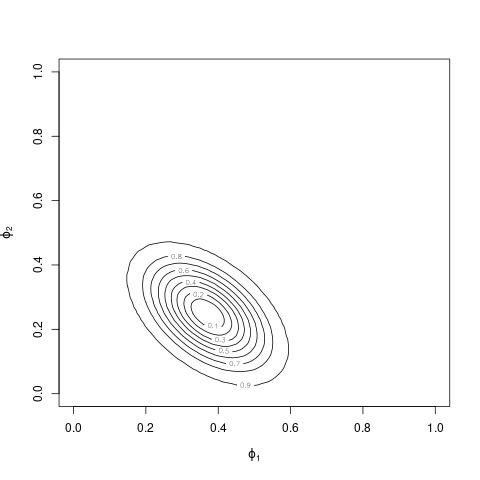}
\end{center}
\caption{Simulated confidence curve obtained by bootstrap of the Wald statistic for the case $n=100$, $\sigma^2=1$, $\boldsymbol\phi_0=(0.4,\ 0.2)$, $\hat\bphi_0\approx (0.37,\ 0.25)$.}
\end{figure}

\begin{figure}[htb]
\begin{center}
\includegraphics[width=90mm]{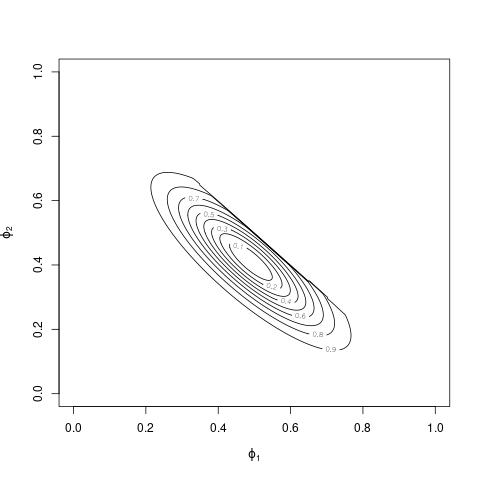}
\end{center}
\caption{Simulated confidence curve obtained by $\chi^2$ approximation of the Wald statistic for the case $n=50$, $\sigma^2=1$, $\boldsymbol\phi_0=(0.6,\ 0.3)$, $\hat\bphi_0\approx (0.49,\ 0.41)$.}
\end{figure}

\begin{figure}[htb]
\begin{center}
\includegraphics[width=90mm]{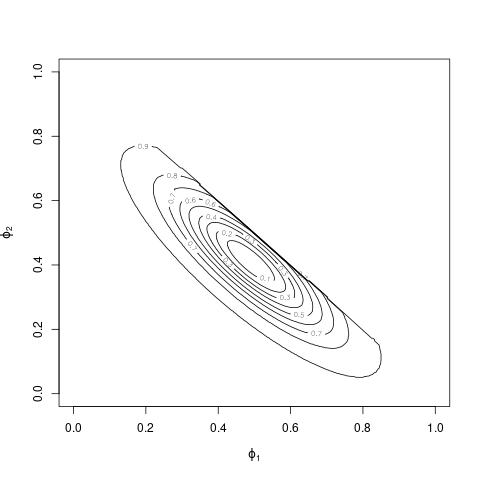}
\end{center}
\caption{Simulated confidence curve obtained by bootstrap of the Wald statistic for the case $n=50$, $\sigma^2=1$, $\boldsymbol\phi_0=(0.6,\ 0.3)$, $\hat\bphi_0\approx (0.49,\ 0.41)$.}
\end{figure}

From Figures 1 and 2, we note that both confidence curves are centered at a point close to the MLE. They look quite similar to each other, although at a closer look, the bootstrap curve appears to be slightly less concentrated around the MLE. It is quite expected that they look similar to each other, since $n=100$ may be considered a fairly large sample size and the true and estimated parameters are inside of the stationarity region by a large margin. As is well-known (cf SS, p. 88), this region is given by the conditions 
\begin{equation}\phi_1+\phi_2<1,\ \phi_2-\phi_1<1,\ \phi_2>-1.
\label{statreg}
\end{equation}

The situation is different in Figures 3 and 4. Here, we are close to the boundary $\phi_1+\phi_2=1$ of the stationarity region. Also, compared to the previous case, we have a smaller sample size here ($n=50$). These two facts imply that the $\chi^2$ approximation works less well. Note that the bootstrap based confidence curve has a much larger 'spread' than its asymptotic counterpart, the latter quite probably being too 'optimistic'. But the most striking feature is the 'cutoff' on the boundary $\phi_1+\phi_2=1$, where in both Figures, the confidence curve immediately rises to 1. This must be so, because non stationary parameter combinations are outruled from the model.  Hence, in this case, confidence regions at realistic levels will not be of purely elliptic shapes.

For more details on simulation setups, the reader is referred to Appendix 2.

\subsection{Confidence density function}

Suppose we have observed the MLE $\hat\bphi_{obs}=(\hat\phi_{1,obs},\hat\phi_{2,obs},...,\hat\phi_{p,obs})'$. In analogy with (\ref{Pvaluem}), the confidence distribution function is found from
\begin{equation}
C(\phi_1,...,\phi_p)=P_{\phi_1,...,\phi_p}(\hat\phi_1\geq\hat\phi_{1,obs},...,\ \hat\phi_p\geq\hat\phi_{p,obs}),\label{Cexact}
\end{equation}
where the probability is calculated under the distribution of the maximum likelihood estimators $\hat\bphi=(\hat\phi_1,...,\hat\phi_p)'$.

This corresponds to the approximated confidence distribution function (cf (\ref{MLas}))
\begin{align}
\tilde C(\bphi)
&=\int_{x_1=\hat\phi_{1,obs}-\phi_1}^\infty
\int_{x_2=\hat\phi_{2,obs}-\phi_2}^\infty\ldots
\int_{x_p=\hat\phi_{p,obs}-\phi_p}^\infty\notag\\
&\frac{n^{p/2}}{(2\pi)^{p/2}|\hat\Omega|^{1/2}}
\exp\left(-\frac{n}{2}{\bf x}'\hat\Omega^{-1}{\bf x}\right)dx_p\ldots dx_2dx_1,\label{Capprox}
\end{align}
where ${\bf x}=(x_1,x_2,...,x_p)'$. An approximate confidence density of $\bphi$ is obtained by successively differentiating (\ref{Capprox}) with respect to $\phi_1,...,\phi_p$. To this end, we note that
\begin{align*}
\frac{\partial}{\partial\phi_1}\tilde C(\bphi)&=
\int_{x_2=\hat\phi_{2,obs}-\phi_2}^\infty\ldots
\int_{x_p=\hat\phi_{p,obs}-\phi_p}^\infty\\
&\frac{n^{p/2}}{(2\pi)^{p/2}|\hat\Omega|^{1/2}}
\exp\left(-\frac{n}{2}\tilde{\bf x}'\hat\Omega^{-1}\tilde{\bf x}\right)dx_p\ldots dx_2,
\end{align*}
where $\tilde{\bf x}=(\hat\phi_{1,obs}-\phi_1,x_2,...,x_p)'$. (The minus sign resulting from inserting the lower integration limit cancels with the minus sign of the derivative of the same limit.)  Going on like this, we obtain the approximate confidence density
\begin{align}
\tilde c(\bphi)&=\frac{\partial^p}{\partial\phi_1\cdots\partial\phi_p}\tilde C(\bphi)\notag\\
&=\frac{n^{p/2}}{(2\pi)^{p/2}|\hat\Omega|^{1/2}}
\exp\left\{-\frac{n}{2}(\hat\bphi_{obs}-\bphi)'\hat\Omega^{-1}(\hat\bphi_{obs}-\bphi)\right\}.\label{cdens}
\end{align}

The finite sample confidence density (cd) may be estimated from the corresponding confidence distribution function (cdf) as follows, where we specialize to the bivariate case only. At first, the cd is estimated using (\ref{Cexact}) and bootstrap. To get an explicit approximation of the cdf from this, we generalize the approach of Larsson (2024) and regress $\Phi^{-1}$ of the cdf on a quadratic function of the parameters $(\phi_1,\phi_2)$, where $\Phi^{-1}(\cdot)$ is the inverse standard normal distribution function. The estimated cd, $\hat c(\phi_1,\phi_2)$ say, is then obtained as the mixed derivative of this approximate cdf. See Appendix 2 for further details.

It is then possible to estimate a $1-\alpha$ confidence region by numerically finding a constant $K$ such that $\{c(\phi_1,\phi_2)>K\}$ defines a region in the $\phi_1\times\phi_2$ plane such that the volume under $c(\phi_1,\phi_2)$ over this region equals $1-\alpha$.

The corresponding operation may be applied to the asymptotic approximation of the cd given in (\ref{cdens}). In Figures 5 and 6, we compare 95\% confidence regions from these methods (called 'cd bootstrap' and 'cd asymptotic') to the Wald based confidence regions which are obtained from the corresponding Wald confidence curves described under section 3.1. For these simulations, we have chosen $n=100$ and $\sigma^2=1$. Figure 5 is based on a simulation with $\boldsymbol\phi_0=(0,\ 0)$, $\hat\bphi_0\approx (0.04,\ 0.07)$ and for Figure 6, we have $\boldsymbol\phi_0=(0.4,\ 0.2)$, $\hat\bphi_0\approx (0.35,\ 0.19)$. The 'sausage shaped' cd bootstrap region in Figure 5 is found to be fairly large compared to the regions for the other three methods. In Figure 6, the cd bootstrap region is a bit of similarly shaped, but comparatively smaller.  (Note that it is prevented to cross the $\phi_1+\phi_2=1$ line.) In Figure 5, the regions based on the asymptotic confidence density and the Wald statistic with asymptotic critical value basically coincide, while the region based on the Wald statistic with a bootstrapped critical value is slightly wider. We observe about the same in Figure 6, but here the difference to bootstrapped Wald is smaller. 

This tells us something about the areas of the confidence regions, but it gives no information of their coverage. We will come back to that in the simulation study of the next section.

\begin{figure}[htb]
\begin{center}
\includegraphics[width=90mm]{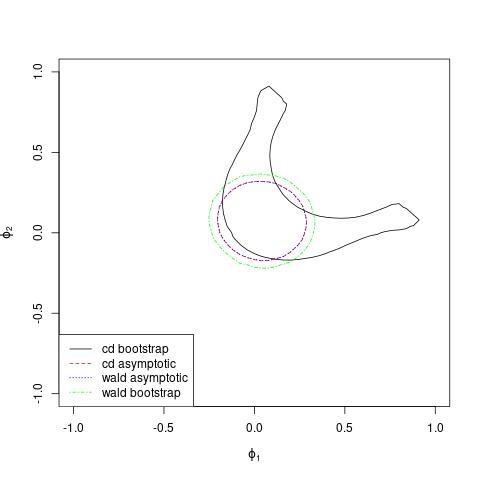}
\end{center}
\caption{Simulated 95\% confidence regions for the case $n=100$, $\sigma^2=1$, $\boldsymbol\phi_0=(0,\ 0)$, $\hat\bphi_0\approx (0.04,\ 0.07)$.}
\end{figure}

\begin{figure}[htb]
\begin{center}
\includegraphics[width=90mm]{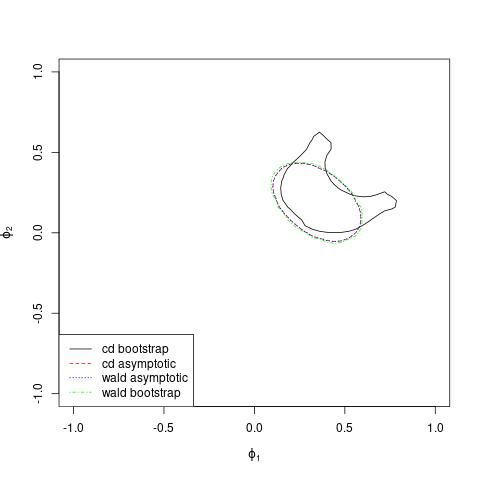}
\end{center}
\caption{Simulated 95\% confidence regions for the case $n=100$, $\sigma^2=1$, $\boldsymbol\phi_0=(0.4,\ 0.2)$, $\hat\bphi_0\approx (0.35,\ 0.19)$.}
\end{figure}

\newpage
  
\subsection{Likelihood}

Write $\bphi=(\phi_1,\phi_2,...,\phi_p)'$, and recall that by assumption,\\ $y_0=y_{-1}=...=y_{1-p}=0$. The likelihood is
\begin{align}
L(\bphi,\sigma^2)&=\frac{1}{\sqrt{2\pi\sigma^2}}\exp\left(-\frac{y_1^2}{2\sigma^2}\right)
\cdot\frac{1}{\sqrt{2\pi\sigma^2}}\exp\left\{-\frac{(y_2-\phi_1 y_1)^2}{2\sigma^2}\right\}\notag\\
&\cdot\ldots\cdot\frac{1}{\sqrt{2\pi\sigma^2}}
\exp\left\{-\frac{(y_p-\phi_1 y_{p-1}-...-\phi_{p-1} y_1)^2}{2\sigma^2}\right\}\notag\\
&\cdot\prod_{t=p+1}^n\frac{1}{\sqrt{2\pi\sigma^2}}
\exp\left\{-\frac{(y_t-\phi_1 y_{t-1}-...-\phi_p y_{t-p})^2}{2\sigma^2}\right\}\notag\\
&=(2\pi\sigma^2)^{-n/2}\exp\left(-\frac{1}{2\sigma^2}A\right),
\label{arlik}
\end{align}
where, defining $\ym=(y_{t-1},y_{t-2},...,y_{t-p})'$,
\begin{align}
A&=\sum_{t=1}^n(y_t-\phi_1 y_{t-1}-...-\phi_p y_{t-p})^2
=\sum_{t=1}^n(y_t-\bphi'\ym)^2\notag\\
&=\sum_{t=1}^n y_t^2-2\bphi'\sum_{t=1}^n\ym y_t+\bphi'\sum_{t=1}^n\ym\ym'\bphi.\label{A}
\end{align}
Now, because of (\ref{reg}) and (\ref{LSE}), the (observed) MLE of $\bphi$ is
\begin{equation}
\hat\bphi_{obs}=\left(\sum_{t=1}^n\ym\ym'\right)^{-1}\sum_{t=1}^n\ym y_t,\label{MLE}
\end{equation}
which implies that (\ref{A}) may be rewritten as
\begin{align}
A&=\sum_{t=1}^n y_t^2-\bphi'\sum_{t=1}^n\ym\ym'(2\hat\bphi_{obs}-\bphi).\label{Ab}
\end{align}

\subsection{Implied non informative prior}

In the following, we will use $O_p(n^{-\alpha})$ to denote a term on the form $n^{-\alpha} V$, where $V$ is a non degenerate random variable. Similarly, an $o_p(n^{-\alpha})$ term tends to zero in probability as $n\to\infty$, when multiplied by $n^{\alpha}$. (For a more precise definitions, see Mann and Wald, 1943). Finally, an $O(n^{-\alpha})$ term tends to a nonzero finite constant as $n\to\infty$, when multiplied by $n^{\alpha}$.

Now, (\ref{cdens}) and (\ref{arlik}) yield the large sample implied prior (cf (\ref{impprior}))
\begin{equation}
\frac{\tilde c(\bphi)}{L(\bphi,\sigma^2)}=\frac{n^{p/2}}{|\hat\Omega|^{1/2}}(2\pi)^{(n-p)/2}(\sigma^2)^{n/2}\exp\left(-\frac{1}{2\sigma^2}B\right),\label{frac}
\end{equation}
where
\begin{equation}
B=n\sigma^2(\hat\bphi_{obs}-\bphi)'\hat\Omega^{-1}(\hat\bphi_{obs}-\bphi)-A.\label{B}
\end{equation}

We have the following result. (For the definition of a causal model, see e.g. SS, p. 85.)

\begin{proposition}
Say that the data are generated from (\ref{ARp}) with\\ $y_0=y_{-1}=...=y_{1-p}=0$ and parameters $\bphi_0=(\phi_{10},...,\phi_{p0})'$ such that the model is causal. Then, as $n\to\infty$,
$$n^{-1}\log\left\{\frac{\tilde c(\bphi)}{L(\bphi,\sigma^2)}\right\}
=\frac{\log{(2\pi\sigma^2)}+1}{2}+\frac{p}{2}\frac{\log n}{n}+O_p(n^{-1/2}).$$
\end{proposition}
\begin{proof}
See Appendix 3.
\end{proof}

The rest term of order $(\log n)/n$ might seem superfluous, but it indicates that the expectation of the right hand side is at least of this order. In the special case $p=1$, Larsson (2024) showed that the $O_p(n^{-1/2})$ rest term is proportional to a random variate with expectation zero. After a correction of a calculation mistake, this rest term is $n^{-1/2}h(\phi_0,\phi)U$, where $U$ is a standard normal variate and
$$h(\phi_0,\phi)=-\frac{\phi_0(1-2\phi_0\phi+\phi^2)}{(1-\phi_0^2)^{3/2}}.$$
Moreover, observe that the result of Proposition 1 also holds when $\sigma^2$ is replaced by its maximum likelihood estimator.

\newpage

\section{Simulation study}

In this section, we present a twofold simulation study. The first part investigates coverage and mean areas of the confidence regions derived by the four methods that were illustrated in section 3.2. The second part empirically checks the convergence of the implied prior to the uniform distribution as the sample size increases.

In the first part of the simulation study, we focus on the sample sizes  $n\in\{50, 100, 200, 400\}$, the true parameter values $\bphi_0=(0,\ 0)$ and\\ $\bphi_0=(0.4,\ 0.2)$ and the confidence levels 90\% and 95\%. The number of replicates is 10000 for $n=50, 100, 200$ and 5000 for $n=400$. For further details, see Appendix 2.

The results are given in Tables 1 and 2 for $\bphi_0=(0,\ 0)$ and\\ $\bphi_0=(0.4,\ 0.2)$, respectively. It is found that the finite sample confidence density method ('cd bootstrap') does not perform so well, as it gives much to low coverage probabilities for small samples, while the mean area is relatively large, in particular for $\bphi_0=(0,\ 0)$. Given Figures 5-6, the latter is not surprising. All the other regions perform reasonably well, and it is striking that the asymptotic confidence density and Wald methods yield almost identical results. Moreover, we may note that for small sample sizes, the Wald bootstrap method is slightly more accurate than the two asymptotic methods. For the asymptotic methods and the Wald bootstrap, as the sample size $n$ increases the empirical coverage approaches the nominal at a rate of about $n^{-1}$. 

\begin{table}
\begin{center}
\begin{tabular}{llllll}
\hline
Sample size&Method&\multicolumn{2}{c}{90\% region}&\multicolumn{2}{c}{95\% region}\\
\hline
&&coverage&mean area&coverage&mean area\\
\cline{3-6}
50&cd bootstrap&0.858&0.507&0.902&0.587\\
&cd asymptotic&0.879&0.270&0.934&0.351\\
&wald asymptotic&0.879&0.270&0.935&0.351\\
&wald bootstrap&0.908&0.309&0.954&0.409\\
\hline
100&cd bootstrap&0.888&0.357&0.931&0.414\\
&cd asymptotic&0.891&0.137&0.940&0.178\\
&wald asymptotic&0.890&0.137&0.940&0.178\\
&wald bootstrap&0.905&0.147&0.952&0.192\\
\hline
200&cd bootstrap&0.904&0.259&0.948&0.300\\
&cd asymptotic&0.898&0.0690&0.947&0.0900\\
&wald asymptotic&0.897&0.0690&0.946&0.0898\\
&wald bootstrap&0.904&0.0714&0.951&0.0932\\
\hline  
400&cd bootstrap&0.879&0.191&0.949&0.220
\\
&cd asymptotic&0.907&0.0348&0.955&0.0450
\\
&wald asymptotic&0.907&0.0346&0.955&0.0451
\\
&wald bootstrap&0.909&0.0352&0.956&0.0459
\\
\hline
\end{tabular}
\vskip.2cm
\caption{Coverage and mean area, $\bphi_0=(0,\ 0)$.}
\end{center}
\end{table}

\vskip.5cm
\begin{table}
\begin{center}
\begin{tabular}{llllll}
\hline
Sample size&Method&\multicolumn{2}{c}{90\% region}&\multicolumn{2}{c}{95\% region}\\
\hline
&&coverage&mean area&coverage&mean area\\
\cline{3-6}
50&cd bootstrap&0.804&0.229&0.852&0.267\\
&cd asymptotic&0.866&0.225&0.922&0.290\\
&wald asymptotic&0.876&0.229&0.931&0.296\\
&wald bootstrap&0.915&0.275&0.961&0.363\\
\hline
100&cd bootstrap&0.850&0.140&0.886&0.162\\
&cd asymptotic&0.889&0.115&0.940&0.149\\
&wald asymptotic&0.889&0.115&0.941&0.150\\
&wald bootstrap&0.912&0.129&0.955&0.169\\
\hline
200&cd bootstrap&0.913&0.0988&0.943&0.114\\
&cd asymptotic&0.900&0.0578&0.949&0.0752\\
&wald asymptotic&0.900&0.0577&0.948&0.0751\\
&wald bootstrap&0.912&0.0615&0.955&0.0803\\
\hline  
400&cd bootstrap&0.928&0.0698&0.954&0.0808
\\
&cd asymptotic&0.897&0.0290&0.948&0.0377
\\
&wald asymptotic&0.898&0.0288&0.948&0.0376
\\
&wald bootstrap&0.903&0.0298&0.952&0.0389
\\
\hline
\end{tabular}
\vskip.2cm
\caption{Coverage and mean area, $\bphi_0=(0.4,\ 0.2)$.}
\end{center}
\end{table}

The second part of the simulation study concerns calculation of the empirical implied prior. In Figures 7-9, for $\bphi=(0,0)$ we give contour plots of
\begin{equation}n^{-1}\log\left\{\frac{\tilde c(\bphi)}{L(\bphi,\sigma^2)}\right\}
-\frac{\log{(2\pi\sigma^2)}+1}{2},\label{logimp}
\end{equation}
for dimension $p=2$. The number of replications is 200 for $n=200, 400$ and 100 for $n=800$. What is depicted is the mean of (\ref{logimp}) over the replications. (See also Appendix 1 where corresponding plots for $\bphi=(0.4,0.2)$ are shown.) We observe that the empirical priors approach their asymptotic limits well as $n$ increases, at a rate which seems to be of an order close to $n^{-1}$ overall. Thus, comparing to proposition 1, it seems that the $O_p(n^{-1/2})$ term vanishes when taking the expectation.

\begin{figure}[htb]
\begin{center}
\includegraphics[width=90mm]{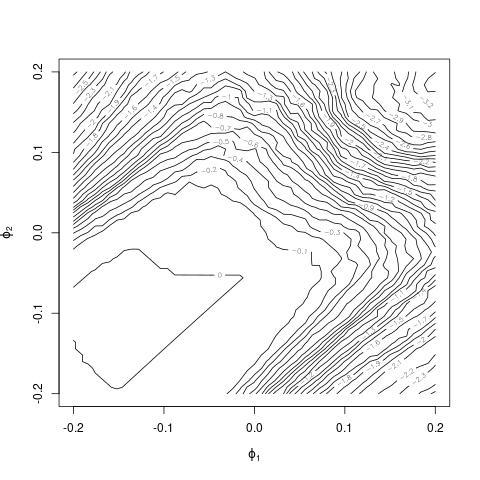}
\end{center}
\caption{Contour plot of the log implied prior minus the main term of proposition 1, $n=200$, $\bphi=(0,0)$, mean over 200 replications.}
\end{figure}

\begin{figure}[htb]
\begin{center}
\includegraphics[width=90mm]{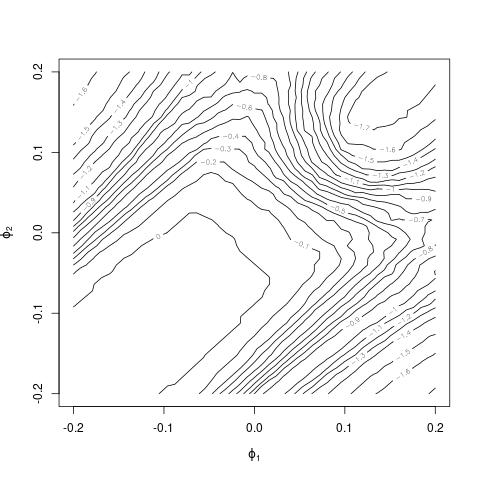}
\end{center}
\caption{Contour plot of the log implied prior minus the main term of proposition 1, $n=400$, $\bphi=(0,0)$, mean over 200 replications.}
\end{figure}

\begin{figure}[htb]
\begin{center}
\includegraphics[width=90mm]{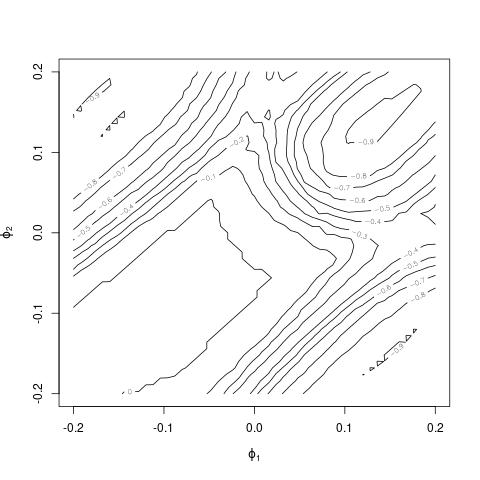}
\end{center}
\caption{Contour plot of the log implied prior minus the main term of proposition 1, $n=800$, $\bphi=(0,0)$, mean over 100 replications.}
\end{figure}

\section{Empirical examples}
Our first example is collected from the textbook Wei (2006), p. 111. It describes the daily average number of defects found at the end of the assembly line of a truck manufacturing plant. The data set contains 45 observations. The series is depicted in Figure 10. 

\begin{figure}[htb]
\begin{center}
\includegraphics[width=90mm]{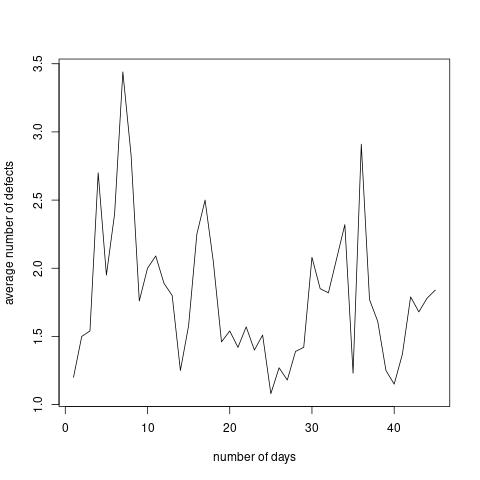}
\end{center}
\caption{Daily average number of defects found at the end of the assembly line of a truck manufacturing plant.}
\end{figure}

In order to get model residuals with a better resemblance to the normal distribution, the original series $y_t$ was transformed according to
$$z_t=\log\left(y_t-\min_t y_t+0.1\right).$$
The model for $z_t$ without intercept with the smallest AIC was found to be an AR(2), and using the standard R procedure, the maximum likelihood estimates of the parameters were $\hat\phi_1=0.4719$ and $\hat\phi_2=0.2262$, with standard errors 0.1470 and 0.1475, respectively. The ADF unit root t test with two lags resulted in a p value of 0.072, thus quite barely not rejecting the null of non stationarity at the 5\% level.

Using the previously discussed four methods, as well as two more methods which will be explained in the sequel, Figure 11 depicts 95\% confidence regions for $\phi_1$ and $\phi_2$. Similar to the illustrations in Figures 5 and 6, we find that the region based on the confidence density and bootstrap differs a bit from those based on the asymptotic confidence density or the Wald statistic (both the one based on $\chi^2$ critical values and the one based on critical values obtained by bootstrap). It is interesting to note that all confidence regions cover parameters on the $\phi_1+\phi_2=1$ line. Since the model is non stationary for such values, this means that non stationarity can not be ruled out. This fact corroborates with the p value $0.072>0.05$ of the unit root test.

The Bayes region of Figure 11 is obtained by assuming a uniform prior for $(\phi_1,\phi_2)$. In general, if $g(\theta)$ is the prior density for the parameter $\theta$ and $L(\theta;{\bf y})$ is the likelihood obtained from the data ${\bf y}$, then the posterior density, $h(\theta)$ say, is obtained from
\begin{equation}
h(\theta)=L(\theta)g(\theta).\label{post}
\end{equation}
Hence, in case the prior $g(\theta)$ is uniform, the posterior density is proportional to the likelihood, and it may be constructed by dividing the likelihood with the integral of the likelihood over the whole parameter space. In our example, the parameter space is defined by those $(\phi_1,\phi_2)$ combinations that make the model causal (i.e. stationary for a suitable choice of initial distribution). This region is given by (\ref{statreg}).
Hence, just as in the case of the asymptotic confidence density region, we obtain the $1-\alpha$ ‘confidence’ (or credibility) Bayes region by finding a constant $K$ such that $\{h(\phi_1,\phi_2)>K\}$ defines a region in the $\phi_1\times\phi_2$ plane such that the volume under $h(\phi_1,\phi_2)$ over this region equals $1-\alpha$, where $h(\phi_1,\phi_2)$ is the likelihood, normed as described above.

Upon inspection of Figure 11, the Bayes region is found to be larger than the other regions, and in fact, it contains all of them. Comparing to the corresponding AR(1) situation (cf Larsson, 2024), it seems that it can be wise to try a prior which is a mixture of a flat prior and a spike at those $(\phi_1,\phi_2)$ that fulfill $\phi_1+\phi_2=1$, i.e. those parameter combinations which, in view of Figure 11, are on the most interesting part of the border of the stationarity region. Say that the posterior based on the likelihood and a flat prior, $h(\bphi)$ say, fulfills
$$\int\int_{\bphi<\bar\bphi}h(\bphi)d\bphi=1-k,$$
where $\bar\bphi=(\bar\phi_1,\bar\phi_2)$ is defined as the set of $\bphi=(\phi_1,\phi_2)$ such that $\phi_1+\phi_2=1$, and $\bphi<\bar\bphi$ means that $\phi_1<\bar\phi_1$ and $\phi_2<\bar\phi_2$. This is equivalent to say that $h(\bphi)$ has a spike of size $b$ at the boundary. Next, with similar notation, taking one of the other confidence densities (e.g. the one that builds on the Wald statistic and its asymptotic distribution) $c(\bphi)$, with a spike of size $b$ at the boundary, we get
$$\int\int_{\bphi<\bar\bphi}c(\bphi)d\bphi=1-b.$$
To adjust for the difference between $b$ and $k$, a simplistic approach is to find a constant $a$ such that for a corrected posterior $\tilde h(\bphi)=ah(\bphi)$, 
$$1-b=\int\int_{\bphi<\bar\bphi}\tilde h(\bphi)d\bphi
=a\int\int_{\bphi<\bar\bphi}h(\bphi)d\bphi=a(1-k),$$
i.e. $a=(1-b)/(1-k)$. 

Now, say that the "uncorrected" Bayes region is given by those $\bphi$ such that $h(\bphi)>K_{1-\alpha}$, where $K_{1-\alpha}$ solves
$$\int\int_{h(\bphi)>K_{1-\alpha}}h(\bphi)d\bphi=1-\alpha,$$
while the "corrected" Bayes region is similarly given by 
those $\bphi$ such that $\tilde h(\bphi)>\tilde K_{1-\alpha}$, where $\tilde K_{1-\alpha}$ solves
$$\int\int_{\tilde h(\bphi)>\tilde K_{1-\alpha}}\tilde h(\bphi)d\bphi=1-\alpha.$$
Then, from $\tilde h(\bphi)=ah(\bphi)$, the latter is equivalently given by those $\bphi$ such that $h(\bphi)>\tilde K_{1-\alpha}/a$, where $\tilde K_{1-\alpha}$ solves
$$\int\int_{h(\bphi)>\tilde K_{1-\alpha}/a} h(\bphi)d\bphi=\frac{1-\alpha}{a}.$$
In case $a>1$ (i.e. $k>b$), the corrected region will become smaller than the uncorrected one, which is reasonable since the uncorrected region in a sense underestimates the size of the "spike" at the boundary. This is also what we see from Figure 11, where the corrected Bayes region comes fairly close to the regions obtained from the asymptotic confidence distribution and the two Wald based regions. 

\begin{figure}[htb]
\begin{center}
\includegraphics[width=90mm]{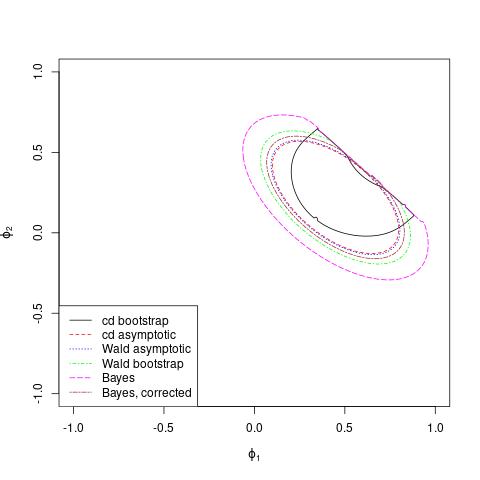}
\end{center}
\caption{95\% confidence regions for the parameters, truck data.}
\end{figure}

Our second example is the R dataset \texttt{LakeHuron}, also given in Brockwell and Davies (2002). This data describes the level of Lake Huron in feet for the years 1875-1972 (98 years). The series is depicted in Figure 12. After subtracting the mean, we fitted an AR(2) model to the data with parameter estimates $\hat\phi_1=1.0441$ and $\hat\phi_2=-0.2503$ and corresponding standard errors 0.0982 and 0.1006, respectively. The ADF unit root test with two lags gave a p value of 0.063, so non stationarity is not rejected at the 5\% level, however close.

\begin{figure}[htb]
\begin{center}
\includegraphics[width=90mm]{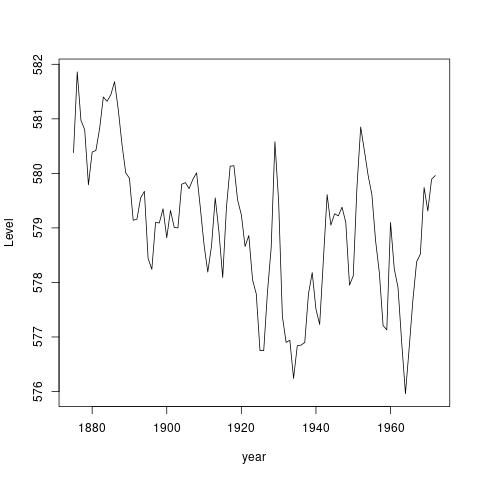}
\end{center}
\caption{Level of Lake Huron in feet for the years 1875-1972.}
\end{figure}

The corresponding confidence regions are displayed in Figure 13. Since the Bayes region did not touch the $\phi_1+\phi_2=1$ line, no correction was worked out. As in the previous example, the (uncorrected) Bayes region contains all the other regions, and among these, the confidence density bootstrap region stands out a little, while the others are very similar.

It is perhaps surprising that, although non stationarity was not rejected, the regions are all inside of the stationarity region. But such things can happen, since the unit root test is evaluated on a different statistical model than the one used for creating confidence regions.

\begin{figure}[htb]
\begin{center}
\includegraphics[width=90mm]{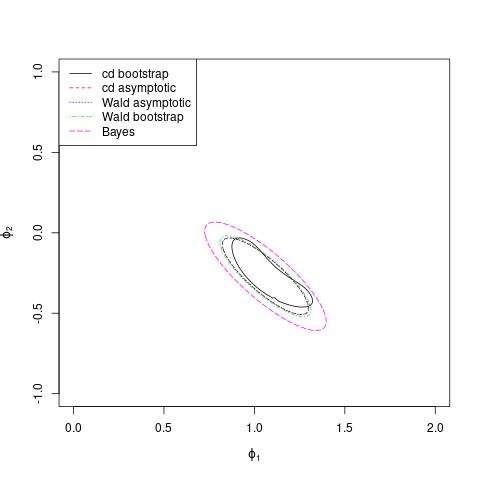}
\end{center}
\caption{95\% confidence regions for the parameters, Lake Huron data.}
\end{figure}

\section{Concluding remarks}

In the present paper, we have suggested a construction of the joint confidence distribution function for several parameters, and applied it to the case of a zero mean autoregressive process of general order $p$. We have proved that asymptotically, the implied prior is flat, thus generalizing the corresponding result in Larsson (2024) for zero mean autoregressive processes of order one. 

Moreover, for the case $p=2$, we have discussed different approaches to calculate confidence curves and confidence regions, and compared them numerically, both in terms of simulations and on real data. These comparisons showed that, in terms of coverage and area of confidence regions, our method of calculating the small sample confidence density in this sense worked less well than a method based on the asymptotic confidence density and methods based on the Wald statistic. The latter three methods show a fairly similar and good performance. 

In future studies, it would be vital to theoretically back up our general construction of confidence distribution functions for several parameters. Also, studying other models than autoregressive ones would be of great interest. 

\section*{Appendix 1: Extra figures}

\begin{figure}[htb]
\begin{center}
\includegraphics[width=90mm]{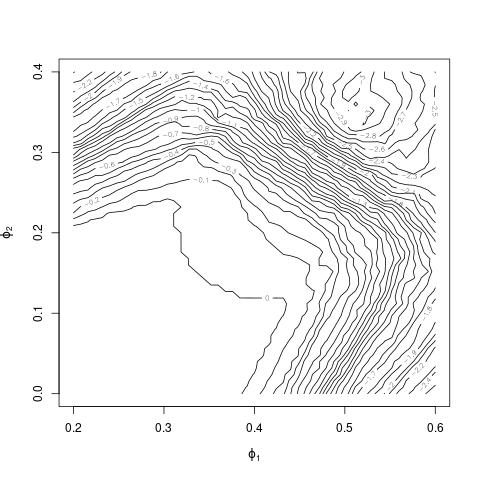}
\end{center}
\caption{Contour plot of the log implied prior minus the main term of proposition 1, $n=200$, $\bphi=(0.4,0.2)$, mean over 200 replications.}
\end{figure}

\begin{figure}[htb]
\begin{center}
\includegraphics[width=90mm]{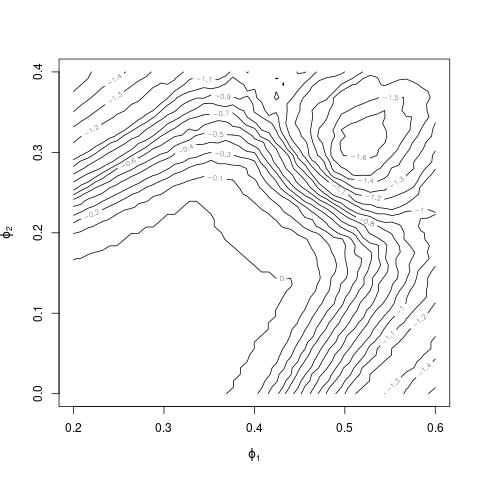}
\end{center}
\caption{Contour plot of the log implied prior minus the main term of proposition 1, $n=400$, $\bphi=(0.4,0.2)$, mean over 200 replications.}
\end{figure}

\begin{figure}[htb]
\begin{center}
\includegraphics[width=90mm]{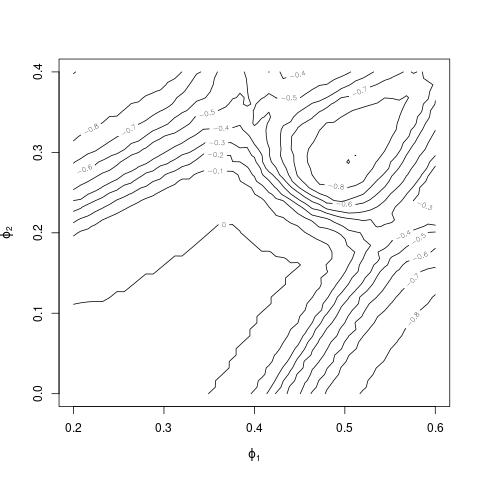}
\end{center}
\caption{Contour plot of the log implied prior minus the main term of proposition 1, $n=800$, $\bphi=(0.4,0.2)$, mean over 100 replications.}
\end{figure}

\section*{Appendix 2: Simulation setups}

In this appendix, we briefly describe the numerical calculations and simulations that give rise to the figures and tables of the paper. In all cases, basically all calculations were performed in Matlab R2019a, and all  plots were done in R, version 4.1.2.

\subsection*{Figures 1-4}

This simulation generates the confidence curves based on the Wald statistic using the asymptotic distribution or bootstrap. For $j=1,2$, the confidence is calculated on a grid defined by $\phi_j=(\phi_{j,min},\phi_{j,min}+h_j,...,\phi_{j,max})$, where $h_j=(\phi_{j,max}-\phi_{j,min})/m$ for some large number $m$. Here, it is important to make sure that the grid only contains such $(\phi_1,\phi_2)$ that satisfy the stationarity conditions.

\begin{enumerate}
\item Given $n$, calculate $h_j$ as given above.
\item For given $\bphi_0$, generate the series and calculate $\hat\bphi_{obs}$.
\item Calculate the bootstrap distribution of the Wald statistic by repeating the following for $k=1,2,...,N$, where $N$ is a large number:
\begin{enumerate}
\item Bootstrap the residuals under the model with $\bphi=\hat\bphi_{obs}$ by drawing them with replacement a large number of times.
\item Given the bootstrapped residuals and $\bphi=\hat\bphi_{obs}$, generate a new series and use that series to calculate the estimate $\hat\bphi_k$ and the corresponding value of the Wald statistic, $Q_k$ say.
\end{enumerate}
The so obtained collection $Q_1,...,Q_N$ constitutes the bootstrap distribution of $Q$.
\item For all 'true parameters' $(\phi_1,\phi_2)$ in the grid, do
\begin{enumerate}
\item Given the 'true parameters' and $\hat\bphi_{obs}$, calculate the Wald statistic $Q_{obs}$.
\item Calculate the asymptotic confidence as $F(Q_ {obs})$, where $F$ is the distribution function of the asymptotic $\chi^2$ distribution.
\item Calculate the bootstrap confidence as the proportion of $Q_k$ that are smaller than $Q_{obs}$.
\end{enumerate}
\item Plot the confidence curves.
\end{enumerate}

\subsection*{Figures 5-6}

The simulations lying behind Figures 5-6 calculate the Wald type confidence regions as above. In addition to that, we need to calculate regions based on the finite sample and asymptotic confidence densities. 

For the asymptotic confidence density based region, the procedure is as above where the asymptotic confidence density is computed for all grid parameters under step 4 above. After that,
\begin{enumerate}
\item The confidence density is normed in order for the total volume under it to equal one. 
\item For all parameters in the grid and given the confidence level $1-\alpha$, by numerical equation solving we calculate $K_{1-\alpha}$ such that the the volume under the part of the confidence density where it is larger than $K_{1-\alpha}$ equals $1-\alpha$. The Matlab command used for this is called \texttt{fzero}.
\item The contour consisting of all $K_{1-\alpha}$ of the previous step may be plotted.
\end{enumerate}

The finite sample confidence region is trickier to obtain. At first, for all grid parameters the confidence distribution function (cdf) is calculated under step 4 (for Figures 1-4) above. For all grid parameters, this is done as follows:
\begin{enumerate}
\item Repeat the following for $k=1,2,...,N$, where $N$ is a large number:
\begin{enumerate}
\item Generate the time series using the 'true' grid parameters\\ $\bphi=(\phi_1,\phi_2)$ and innovations simulated from the standard normal distribution.
\item Calculate $\hat\bphi=(\hat\phi_1,\hat\phi_2)$ from the generated series.
\item Check if  $\hat\phi_1>\phi_1$ and $\hat\phi_2>\phi_2$.
\end{enumerate}
\item The cdf at $\bphi=(\phi_1,\phi_2)$ is estimated as the proportion of times that $\hat\phi_1>\phi_1$ and $\hat\phi_2>\phi_2$ above.
\end{enumerate}

We now have the cdf for all grid parameters, and from that, we want a numerical estimate of the confidence density (cd). This is a non trivial problem, see e.g. Knowles and Renka (2014). We have chosen to generalize the idea of Larsson (2024) to fit an explicit function to the cdf and then obtain the cd by analytical differentiation. This is done as follows:

\begin{enumerate}
\item For all $(\phi_1,\phi_2)$ in the grid, calculate $\Phi^{-1}\{C(\phi_1,\phi_2)\}$, where $C(\cdot,\cdot)$ is the cdf calculated above and $\Phi^{-1}$ is the inverse of the standard normal distribution function.
\item Estimate the regression 
$$\Phi^{-1}\{C(\phi_1,\phi_2)\}=c_0+c_1\phi_1+c_2\phi_2+c_{11}\phi_1^2+c_{22}\phi_2^2+c_{12}\phi_1\phi_2+\epsilon.$$
\item Calling the estimated parameters $\hat c_0$ etcetera, estimate the cd as
$$\frac{\partial^2}{\partial\phi_1\partial\phi_2}
\Phi(z)=\{\hat c_{12}-(\hat c_1+2\hat c_{11}\phi_1+\hat c_{12}\phi_2)(\hat c_2+2\hat c_{22}\phi_2+\hat c_{12}\phi_1)z\}\varphi(z),$$
where
$$z=\hat c_0+\hat c_1\phi_1+\hat c_2\phi_2+\hat c_{11}\phi_1^2+\hat c_{22}\phi_2^2+\hat c_{12}\phi_1\phi_2$$
and $\varphi(\cdot)$ is the standard normal density function.
\end{enumerate}

We also tried several other regressions, e.g. of the same type as here but including higher powers of the parameters as regressors, but the results did not improve.

Another remarks is that, to avoid numerical problems, when estimating the regression only those $\phi_1,\phi_2$ that fulfilled $\delta<C(\phi_1,\phi_2)<1-\delta$ are used, where $\delta$ is the minimum of 0.1 and
$$\exp(6.04-2.64\log n+4.39\phi_1+9.92\phi_2).$$
The motivation for this choice is a regression of the squared volume under the empirical cd on $\log\delta$ and $\log n$ for some choice of $n$, $\phi_1$ and $\phi_2$ values, and then solving for the parameters for a square volume of 2. (Not 1, as one should think, but by trial and error, this seems to be a more robust choice.)

\subsection*{Tables 1-2}

The code generating Tables 1-2 builds on the codes above. For each replicate in an outer loop, it adds area calculations and checks that the true parameters are inside the confidence regions. Then, the mean is taken over the replicates in this loop.

\subsection*{Figures 7-9, 14-16}

For each replicate in an outer loop, the code for Figures 7-9, 14-16 calculates the confidence density as for Figures 5-6 above. It also calculates the likelihood, and then the log of their ratio. The mean is taken over all replicates in the outer loop.

A key thing here is that the confidence density can become too close to zero, and then taking the logarithm does not work. So it has to be truncated below at some value, in our case $10^{-323}$. This truncation limit has to be taken as small as possible, since preliminary studies show that this choice can affect the implied prior quite much for parameters in the grid far from the generating ones.

\subsection*{Figures 11, 13}

The code for Figures 11 and 13 is as for the Figures 5-6 code, differing only in that the empirical series is used as input, rather than a simulated one.

\section*{Appendix 3: Proof of proposition 1}

Let the vector of true parameters be $\bphi_0$. By e.g. SS, p. 125, we have that as $n\to\infty$, denoting convergence in distribution by $\dlim$,  
$$\sqrt{n}(\hat\bphi_{obs}-\bphi_0)\dlim {\bu},$$
for a normally distributed variate $\bu$ with expectation zero and covariance matrix $\Omega$. We write this as
\begin{equation}
\hat\bphi_{obs}=\bphi_0+O_p(n^{-1/2}).\label{phihat}
\end{equation}
Now, (\ref{phihat}) implies
\begin{align}
&(\hat\bphi_{obs}-\bphi)'\hat\Omega^{-1}(\hat\bphi_{obs}-\bphi)\notag\\
&=\left\{(\hat\bphi_{obs}-\bphi_0)+(\bphi_0-\bphi)\right\}'
\hat\Omega^{-1}\left\{(\hat\bphi_{obs}-\bphi_0)+(\bphi_0-\bphi)\right\}
\notag\\
&=(\bphi_0-\bphi)'\hat\Omega^{-1}(\bphi_0-\bphi)
+O_p(n^{-1/2}).
\label{phiobsphi}
\end{align}
Similarly, we have
\begin{equation}
2\hat\bphi_{obs}-\bphi=2\bphi_0-\bphi+O_p(n^{-1/2}).\label{phiobsphi2}
\end{equation}
Moreover, define (let $q=p-1$) 
$$\Phi_0=\left(\begin{array}{cc}
\begin{array}{cccc}\phi_{10}&\phi_{20}&\ldots&\phi_{q0}\end{array}
&\phi_{p0}\\{\bf I}_q&{\bf 0}_q\end{array}\right),
\quad \te_t=\left(\begin{array}{c}\varepsilon_t\\{\bf 0}_q\end{array}\right),$$
where for all integer $q$, ${\bf I}_q$ is the identity matrix and ${\bf 0}_q$ is the zero vector of dimension $q$.
We have by recursion that (put $j=t-1-i$ in the last step)
\begin{align}
\ym&=\Phi_0{\bf y}_{t-1,-}+\te_{t-1}
=\Phi_0\left(\Phi_0{\bf y}_{t-2,-}+\te_{t-2}\right)+\te_{t-1}\notag\\
&=\Phi_0^2{\bf y}_{t-2,-}+\Phi_0\te_{t-2}+\te_{t-1}=\ldots=
\sum_{i=0}^{t-2}\Phi_0^i\te_{t-1-i}\notag\\
&=\sum_{j=1}^{t-1}\Phi_0^{t-1-j}\te_j,\label{ym}
\end{align}
implying
\begin{equation}
\sum_{t=1}^n\ym\ym'=\sum_{t=1}^n\sum_{i=1}^{t-1}\sum_{j=1}^{t-1}\Phi_0^{t-1-i}\te_i(\te_j)'(\Phi_0')^{t-1-j}.\label{sumy}
\end{equation}
Observing that $E\{\te_i(\te_j)'\}=\sigma^2{\bf M}$ if $i=j$ and zero otherwise, where (${\bf 0}_{q\times q}$ is a $q\times q$ zero matrix)
$${\bf M}=\left(\begin{array}{cc}1&{\bf 0}_q'\\{\bf 0}_q&{\bf 0}_{q\times q}\end{array}\right),$$
(\ref{sumy}) yields
$$E\left(\sum_{t=1}^n\ym\ym'\right)=
\sigma^2\sum_{t=1}^n\sum_{i=1}^{t-1}\Phi_0^{t-1-i}{\bf M}(\Phi_0')^{t-1-i},$$
and utilizing the identity
\begin{equation}
\ve({\bf A}{\bf B}{\bf C})=({\bf C}'\otimes{\bf A})\ve({\bf B}),
\label{vecid}
\end{equation}
for matrices ${\bf A}$, ${\bf B}$, ${\bf C}$ of dimensions such that all products are defined, and putting $\bPhi=\Phi_0\otimes\Phi_0$, we get
\begin{align}
E\left\{\ve\left(\sum_{t=1}^n\ym\ym'\right)\right\}
&=\sigma^2\sum_{t=1}^n\sum_{i=1}^{t-1}\ve\left\{\Phi_0^{t-1-i}{\bf M}(\Phi_0')^{t-1-i}\right\}\notag\\
&=\sigma^2\sum_{t=1}^n\sum_{i=1}^{t-1}(\Phi_0^{t-1-i}\otimes\Phi_0^{t-1-i})\ve({\bf M})\notag\\
&=\sigma^2\sum_{t=1}^n\sum_{i=1}^{t-1}\bPhi^{t-1-i}\ve({\bf M})\notag\\
&=\sigma^2\sum_{i=1}^{n-1}\left(\sum_{t=i+1}^n\bPhi^t\right)\bPhi^{-1-i}\ve({\bf M}).\label{Evsumy}
\end{align}
Next, from the identity
$$\sum_{t=u}^v{\bf C}^t=({\bf I}-{\bf C})^{-1}({\bf C}^u-{\bf C}^{v+1}),$$
where ${\bf C}$ is a matrix such that ${\bf I}-{\bf C}$ is non singular, and where ${\bf I}$ is the identity matrix, (\ref{Evsumy}) implies
\begin{align}
&E\left\{\ve\left(\sum_{t=1}^n\ym\ym'\right)\right\}\notag\\
&=\sigma^2({\bf I}_{p^2}-\bPhi)^{-1}\sum_{i=1}^{n-1}
(\bPhi^{i+1}-\bPhi^{n+1})\bPhi^{-1-i}\ve({\bf M})\notag\\
&=\sigma^2({\bf I}_{p^2}-\bPhi)^{-1}\left\{(n-1){\bf I}_{p^2}-
\sum_{i=1}^{n-1}\bPhi^{n-i}\right\}\ve({\bf M})\notag\\
&=n\sigma^2({\bf I}_{p^2}-\bPhi)^{-1}\ve({\bf M})+O(1).\label{Evsumy2}
\end{align}

As for the corresponding variance, (\ref{sumy}) implies
\begin{align*}
&\left(\sum_{s=1}^n\ym\ym'\right)\left(\sum_{t=1}^n\ym\ym'\right)'\\
&=\sum_{s=1}^n\sum_{t=1}^n\sum_{i=1}^{s-1}\sum_{j=1}^{s-1}\sum_{k=1}^{t-1}\sum_{l=1}^{t-1}\Phi_0^{s-1-i}\te_i(\te_j)'(\Phi_0')^{s-1-j}
\Phi_0^{t-1-k}\te_k(\te_l)'(\Phi_0')^{t-1-l},
\end{align*}
and applying (\ref{vecid}) with ${\bf B}={\bf I}_p$, we obtain
\begin{align}
&\ve\left\{\left(\sum_{s=1}^n\ym\ym'\right)\left(\sum_{t=1}^n\ym\ym'\right)'\right\}\notag\\
&=\sum_{s=1}^n\sum_{t=1}^n\sum_{i=1}^{s-1}\sum_{j=1}^{s-1}\sum_{k=1}^{t-1}\sum_{l=1}^{t-1}\notag\\
&\left[\left\{\Phi_0^{t-1-l}\te_l(\te_k)'(\Phi_0')^{t-1-k}\right\}\otimes
\left\{\Phi_0^{s-1-i}\te_i(\te_j)'(\Phi_0')^{s-1-j}\right\}\right]\ve({\bf I}_p)
\notag\\
&=\sum_{s=1}^n\sum_{t=1}^n\sum_{i=1}^{s-1}\sum_{j=1}^{s-1}\sum_{k=1}^{t-1}\sum_{l=1}^{t-1}\left[\left(\Phi_0^{t-1-l}\otimes\Phi_0^{s-1-i}\right)
\left\{\te_l(\te_k)'\otimes\te_i(\te_j)'\right\}\right.\notag\\
&\left.\left\{(\Phi_0')^{t-1-k}\otimes(\Phi_0')^{s-1-j}\right\}\right]\ve({\bf I}_p).
\label{vsumyy}
\end{align}
To take the expectation in (\ref{vsumyy}), we find that since
\begin{align*}
\te_l(\te_k)'\otimes\te_i(\te_j)'&=
\left(\begin{array}{c}\varepsilon_l\\{\bf 0}_q\end{array}\right)
\left(\begin{array}{cc}\varepsilon_k&{\bf 0}_q'\end{array}\right)\otimes
\left(\begin{array}{c}\varepsilon_i\\{\bf 0}_q\end{array}\right)
\left(\begin{array}{cc}\varepsilon_j&{\bf 0}_q'\end{array}\right)\\
&=\left(\begin{array}{cc}\varepsilon_l\varepsilon_k&{\bf 0}_q'\\{\bf 0}_q&{\bf 0}_{q\times q}\end{array}\right)\otimes
\left(\begin{array}{cc}\varepsilon_l\varepsilon_k&{\bf 0}_q'\\{\bf 0}_q&{\bf 0}_{q\times q}\end{array}\right)
\\
&=\varepsilon_l\varepsilon_k\varepsilon_i\varepsilon_j{\bf M}^{\otimes 2},
\end{align*}
where we use the notation ${\bf M}^{\otimes 2}={\bf M}\otimes{\bf M}$,
we have
$$E\left\{\te_l(\te_k)'\otimes\te_i(\te_j)'\right\}
=\left\{\begin{array}{cl}3\sigma^4{\bf M}^{\otimes 2},&\text{if}\ i=j=l=k,\\
\sigma^4{\bf M}^{\otimes 2},&\text{if}\ i,j,k,l\ \text{are pairwise equal,}\\
0&\text{otherwise.}\end{array}\right.$$
This means that 
\begin{align}
&E\left[\ve\left\{\left(\sum_{s=1}^n\ym\ym'\right)\left(\sum_{t=1}^n\ym\ym'\right)'\right\}\right]\notag\\
&=\sigma^4(3S_1+S_2+2S_3)\ve({\bf I}_p),\label{vsumyy2}
\end{align}
where, denoting the minimum of $s$ and $t$ by $s\wedge t$,
\begin{align}
&S_1\notag\\
&=\sum_{s=1}^n\sum_{t=1}^n\sum_{i=1}^{(s\wedge t)-1}
\left(\Phi_0^{t-1-i}\otimes\Phi_0^{s-1-i}\right){\bf M}^{\otimes 2}
\left\{(\Phi_0')^{t-1-i}\otimes(\Phi_0')^{s-1-i}\right\},\label{S1}
\end{align}
\begin{align}
S_2
&=\sum_{s=1}^n\sum_{t=1}^n\sum_{i=1}^{s-1}\sum_{k=1}^{t-1}
\notag\\&\left(\Phi_0^{t-1-k}\otimes\Phi_0^{s-1-i}\right){\bf M}^{\otimes 2}
\left\{(\Phi_0')^{t-1-k}\otimes(\Phi_0')^{s-1-i}\right\},\label{S2}
\end{align}
and
\begin{align}
S_3
&=\sum_{s=1}^n\sum_{t=1}^n\sum_{i=1}^{(s\wedge t)-1}\sum_{j=1}^{(s\wedge t)-1}\notag\\
&\left(\Phi_0^{t-1-j}\otimes\Phi_0^{s-1-i}\right){\bf M}^{\otimes 2}
\left\{(\Phi_0')^{t-1-i}\otimes(\Phi_0')^{s-1-j}\right\}.\label{S3}
\end{align}
We start by looking at $S_2$. In the usual manner, we get
\begin{align}
S_2&=\sum_{s=1}^n\sum_{t=1}^n\sum_{i=1}^{s-1}\sum_{k=1}^{t-1}
\left\{(\Phi_0^{t-1-k}{\bf M}(\Phi_0')^{t-1-k}\right\}\otimes
\left\{(\Phi_0^{s-1-i}{\bf M}(\Phi_0')^{s-1-i}\right\}\notag\\
&=\left\{\sum_{t=1}^n\sum_{k=1}^{t-1}\Phi_0^{t-1-k}{\bf M}(\Phi_0')^{t-1-k}\right\}^{\otimes 2},\label{S2b}
\end{align}
where, denoting the inverse vec operator by $\ve^{-1}$,
\begin{align}
&\sum_{t=1}^n\sum_{k=1}^{t-1}\Phi_0^{t-1-k}{\bf M}(\Phi_0')^{t-1-k}
\notag\\
&=\ve^{-1}\left[\sum_{t=1}^n\sum_{k=1}^{t-1}
\ve\left\{\Phi_0^{t-1-k}{\bf M}(\Phi_0')^{t-1-k}\right\}\right]\notag\\
&=\ve^{-1}\left\{\sum_{t=1}^n\sum_{k=1}^{t-1}\bPhi^{t-1-k}\ve({\bf M})\right\}=\ve^{-1}\left\{\sum_{k=1}^{n-1}\sum_{t=k+1}^n\bPhi^{t-1-k}\ve({\bf M})\right\}\notag\\
&=\ve^{-1}\left\{\sum_{k=1}^{n-1}\sum_{u=0}^{n-1-k}\bPhi^u\ve({\bf M})\right\}\notag\\
&=\ve^{-1}\left\{({\bf I}_{p^2}-\bPhi)^{-1}\sum_{k=1}^{n-1}({\bf I}_p-\bPhi^{n-k})\ve({\bf M})\right\}\notag\\
&=n\ve^{-1}\left\{({\bf I}_{p^2}-\bPhi)^{-1}\ve({\bf M})\right\}+O(1).
\label{sumPhiMPhi}
\end{align}
Hence, via (\ref{S2b}),
\begin{equation}
S_2=n^2\left[\ve^{-1}\left\{({\bf I}_{p^2}-\bPhi)^{-1}\ve({\bf M})\right\}\right]
^{\otimes 2}+O(n).\label{S2c}
\end{equation}

Next, we will show that $S_1$ and $S_3$ are at most $O(n)$. Let $||{\bf \cdot}||$ be some matrix norm. We will utilize the equality 
$||{\bf A}\otimes{\bf B}||=||{\bf A}||||{\bf B}||$ (see Lancaster and Farahat, 1972, theorem 8) and the well-known triangle inequality\\ 
$||{\bf A}+{\bf B}||\leq||{\bf A}||+||{\bf B}||$.

As for $S_1$, (\ref{S1}) yields, writing $z=||\Phi_0||$ and using that\\ 
$s+t-2(s\wedge t)=|s-t|$,
\begin{align}
&||S_1||\notag\\
&\leq\sum_{s=1}^n\sum_{t=1}^n\sum_{i=1}^{(s\wedge t)-1}
\left|\left|\left(\Phi_0^{t-1-i}\otimes\Phi_0^{s-1-i}\right){\bf M}^{\otimes 2}
\left\{(\Phi_0')^{t-1-i}\otimes(\Phi_0')^{s-1-i}\right\}\right|\right|\notag\\
&=||{\bf M}||^2\sum_{s=1}^n\sum_{t=1}^n\sum_{i=1}^{(s\wedge t)-1}
z^{2(s+t)-4(i+1)}\notag\\
&=||{\bf M}||^2\frac{z^{-4}}{z^{-4}-1}\sum_{s=1}^n\sum_{t=1}^n
z^{2(s+t)}\left(z^{-4(s\wedge t)}-z^{-4}\right)\notag\\
&=||{\bf M}||^2\frac{1}{1-z^{4}}\left(\sum_{s=1}^n\sum_{t=1}^n
z^{2|s-t|}-\sum_{s=1}^n\sum_{t=1}^nz^{2(s+t)-4}\right).
\label{S1n}
\end{align}
Here, because $|z|<1$,
\begin{equation}
\sum_{s=1}^n\sum_{t=1}^nz^{2|s-t|}
=n+2\sum_{s=1}^n\sum_{t=1}^{s-1}z^{2(s-t)}=n+O(1),\label{sumz1}
\end{equation}
and
\begin{equation}
\sum_{s=1}^n\sum_{t=1}^nz^{2(s+t)-4}=O(1),\label{sumz2}
\end{equation}
which via (\ref{S1n}) yield that $||S_1||\leq O(n)$.

Regarding $S_3$, we similarly have from (\ref{S3}) that
\begin{align}
&||S_3||\notag\\
&\leq\sum_{s=1}^n\sum_{t=1}^n\sum_{i=1}^{(s\wedge t)-1}
\sum_{j=1}^{(s\wedge t)-1}
\left|\left|\left(\Phi_0^{t-1-j}\otimes\Phi_0^{s-1-i}\right){\bf M}^{\otimes 2}
\left\{(\Phi_0')^{t-1-i}\otimes(\Phi_0')^{s-1-j}\right\}\right|\right|\notag\\
&=||{\bf M}||^2\sum_{s=1}^n\sum_{t=1}^n\sum_{i=1}^{(s\wedge t)-1}
\sum_{j=1}^{(s\wedge t)-1}z^{2(s+t)-2(i+j+2)}\notag\\
&=||{\bf M}||^2z^{-4}\sum_{s=1}^n\sum_{t=1}^n
z^{2(s+t)}D_{st},\label{S3n}
\end{align}
where
\begin{align*}
D_{st}&=\sum_{i=1}^{(s\wedge t)-1}z^{-2i}
\sum_{j=1}^{(s\wedge t)-1}z^{-2j}
=\left(\frac{z^{-2(s\wedge t)}-z^{-2}}{z^{-2}-1}\right)^2,
\end{align*}
which inserted into (\ref{S3n}) gives
\begin{align*}
||S_3||&\leq||{\bf M}||^2\frac{1}{(1-z^2)^2}
\left(\sum_{s=1}^n\sum_{t=1}^n z^{2(s+t)-4(s\wedge t)}
-2z^{-2}\sum_{s=1}^n\sum_{t=1}^nz^{2(s+t)-2(s\wedge t)}
\right.\\
&\left.+z^{-4}\sum_{s=1}^n\sum_{t=1}^n
z^{2(s+t)}\right),
\end{align*}
which, as in (\ref{sumz1}) and (\ref{sumz2}), is $O(n)$.

Let ${\bf Y}=\sum_{t=1}^n\ym\ym'$ and ${\bf m}=({\bf I}_p-\bPhi)^{-1}\ve({\bf M})$. We have seen from (\ref{Evsumy2}) that
\begin{equation}
\ve\left\{E({\bf Y})\right\}=n\sigma^2{\bf m}+O(1).\label{EvY}
\end{equation}
Moreover, via (\ref{vsumyy2}) and (\ref{S2c}),
\begin{equation}
\ve\left\{E({\bf Y}{\bf Y}')\right\}=n^2\sigma^4\left\{\ve^{-1}({\bf m})\right\}^{\otimes 2}\ve({\bf I}_p)+O(n).\label{EvYY}
\end{equation}
With $V({\bf Y})$ as the covariance matrix of ${\bf Y}$, we have that
\begin{equation}
\ve\{V({\bf Y})\}=\ve\left\{E({\bf Y}{\bf Y}')\right\}-
\ve\left[\{E({\bf Y})\}\{E({\bf Y})\}'\right].\label{VY}
\end{equation}
We want to deduce that $V({\bf Y})=O(n)$. But this follows from (\ref{vecid}), which via (\ref{EvYY}) implies that
\begin{equation}
\ve\left\{E({\bf Y}{\bf Y}')\right\}=n^2\sigma^4
\ve\left[\ve^{-1}({\bf m})\{\ve^{-1}({\bf m})\}'\right]+O(n).\label{EvYY2}
\end{equation}
Moreover, via (\ref{EvY}),
\begin{equation}
\ve\left[\{E({\bf Y})\}\{E({\bf Y})\}'\right]
=n^2\sigma^4\ve\left[\ve^{-1}({\bf m})\{\ve^{-1}({\bf m})\}'\right]+O(n),
\label{vEY2}
\end{equation}
and combining (\ref{VY})-(\ref{vEY2}), we have $V({\bf Y})=O(n).$
Then, by the Chebychev inequality (cf  Gnedenko (1989), p. 248), as $n\to\infty$, $n^{-1}{\bf Y}$ converges in probability to its asymptotic expectation, which from (\ref{Evsumy2}) is given by
\begin{equation}
E\left\{n^{-1}\ve\left(\sum_{t=1}^n\ym\ym'\right)\right\}
=\sigma^2({\bf I}_{p^2}-\bPhi)^{-1}\ve({\bf M})+O(n^{-1}).\label{Evsumy3}
\end{equation}

Similarly, from (\ref{Evsumy2}),
\begin{equation}
n^{-1}\sum_{t=1}^n y_t^2=\sigma^2\ve({\bf M})'({\bf I}_{p^2}-\bPhi)^{-1}\ve({\bf M})+O(n^{-1}).\label{psumy2}
\end{equation}
Now, combining (\ref{Ab}), (\ref{phiobsphi2}), (\ref{Evsumy3}) and (\ref{psumy2}), we have
\begin{align}
n^{-1}A&=\sigma^2\left[\ve({\bf M})'({\bf I}_{p^2}-\bPhi)^{-1}\ve({\bf M})\right.\notag\\
&\left.-\bphi'\ve^{-1}\left\{({\bf I}_{p^2}-\bPhi)^{-1}\ve({\bf M})\right\}(2\bphi_0-\bphi)\right]+O_p(n^{-1/2}).\label{pA}
\end{align}
Then, via (\ref{B}), (\ref{phiobsphi}) and the fact that $\hat\Omega=\Omega_0+O_p(n^{-1/2})$, we get
\begin{align}
\frac{1}{n\sigma^2}B&=(\bphi_0-\bphi)'\Omega_0^{-1}(\bphi_0-\bphi)\notag\\
&-\ve({\bf M})'({\bf I}_{p^2}-\bPhi)^{-1}\ve({\bf M})\notag\\
&+\bphi'\ve^{-1}\left\{({\bf I}_{p^2}-\bPhi)^{-1}\ve({\bf M})\right\}(2\bphi_0-\bphi)+O_p(n^{-1/2}).\label{Bn}
\end{align}

Below, we will show that the main term of (\ref{Bn}) is no function of $\bphi$ or $\bphi_0$.
To this end, writing (let $\bphi^-_{p0}=(\phi_{20},...,\phi_{p0})$)
$$\Phi_0=\left(\begin{array}{cc}\phi_{10}&\bphi^-_{p0}\\
{\boldsymbol\delta}_q&{\bf J}_q\end{array}\right),
\quad{\boldsymbol\delta}_q=\left(\begin{array}{c}1\\{\bf 0}_{q-1}\end{array}
\right),\quad {\bf J}_q=\left(\begin{array}{cc}{\bf 0}_{q-1}'&0\\
{\bf I}_{q-1}&{\bf 0}_{q-1}\end{array}\right),$$
we have
$$\bPhi=\Phi_0^{\otimes 2}
=\left(\begin{array}{cc}
\phi_{10}\Phi_0&\bphi^-_{p0}\otimes\Phi_0\\
{\boldsymbol\delta}_q\otimes\Phi_0&{\bf J}_q\otimes\Phi_0\end{array}\right).$$
Hence, 
\begin{equation}
{\bf I}_{p^2}-\bPhi
=\left(\begin{array}{cc}
{\bf I}_p-\phi_{10}\Phi_0&-\bphi^-_{p0}\otimes\Phi_0\\
-{\boldsymbol\delta}_q\otimes\Phi_0&{\bf I}_{pq}-{\bf J}_q\otimes\Phi_0\end{array}\right)
=\left(\begin{array}{cc}{\bf A}_{11}&{\bf A}_{12}(q)\\
{\bf A}_{21}(q)&{\bf A}_{22}(q)\end{array}\right),\label{IPhi}
\end{equation}
where ${\bf A}_{11}={\bf I}_p-\phi_{10}\Phi_0$, etcetera. Correspondingly, writing
\begin{equation}
({\bf I}_{p^2}-\bPhi)^{-1}=\left(\begin{array}{cc}{\bf A}^{11}&{\bf A}^{12}(q)\\
{\bf A}^{21}(q)&{\bf A}^{22}(q)\end{array}\right),\label{IPhiinv}
\end{equation}
we have from standard formulae for the partitioned inverse (see e.g. Mardia et al 1979, p.459)
\begin{equation}
{\bf A}^{11}=\left\{{\bf A}_{11}-{\bf A}_{12}(q){\bf A}_{22}(q)^{-1}{\bf A}_{21}(q)\right\}^{-1}.\label{a11inv}
\end{equation}

How can we find ${\bf A}_{22}(q)^{-1}$? Note that
$${\bf J}_q\otimes\Phi_0=\left(\begin{array}{cc}{\bf 0}_{p\times p(q-1)}&
{\bf 0}_{p\times p}\\{\bf I}_{q-1}\otimes\Phi_0&{\bf 0}_{p(q-1)\times p}
\end{array}\right),$$
and writing 
$${\bf I}_{q-1}\otimes\Phi_0=\left(\begin{array}{cc}{\bf I}_{q-2}&{\bf 0}_{q-2}\\{\bf 0}_{q-2}'&1\end{array}\right)\otimes\Phi_0
=\left(\begin{array}{cc}
{\bf I}_{q-2}\otimes\Phi_0&{\bf 0}_{p(q-2)\times p}\\{\bf 0}_{p\times p(q-2)}&\Phi_0
\end{array}\right),$$
we get
$${\bf J}_q\otimes\Phi_0=\left(\begin{array}{ccc}{\bf 0}_{p\times p(q-2)}&
{\bf 0}_{p\times p}&{\bf 0}_{p\times p}\\
{\bf I}_{q-2}\otimes\Phi_0&{\bf 0}_{p(q-2)\times p}&{\bf 0}_{p(q-2)\times p}\\
{\bf 0}_{p\times p(q-2)}&\Phi_0&{\bf 0}_{p\times p}\end{array}\right)
=\left(\begin{array}{cc}{\bf J}_{q-1}&{\bf 0}_{q-1}\\
{\bf K}_{q-1}'&0\end{array}\right)\otimes\Phi_0,$$
where
$${\bf J}_{q-1}=\left(\begin{array}{cc}{\bf 0}_{q-2}'&0\\
{\bf I}_{q-2}&{\bf 0}_{q-2}\end{array}\right),\quad
{\bf K}_{q-1}'=\left(\begin{array}{cc}{\bf 0}_{q-2}'&1\end{array}\right).
$$
Then, in view of (\ref{IPhi}), we may write
$${\bf A}_{22}(q)={\bf I}_{pq}-{\bf J}_q\otimes\Phi_0
=\left(\begin{array}{cc}{\bf A}_{22}(q-1)&{\bf 0}_{p(q-1)\times p}\\
-{\bf K}_{q-1}'\otimes\Phi_0&{\bf I}_p\end{array}\right),$$
from which it follows from (\ref{a11inv}) that (or as is easily verified by multiplication) 
\begin{equation}
{\bf A}_{22}(q)^{-1}=\left(\begin{array}{cc}
{\bf A}_{22}(q-1)^{-1}&{\bf 0}_{p(q-1)\times p}\\
({\bf K}_{q-1}'\otimes\Phi_0){\bf A}_{22}(q-1)^{-1}&{\bf I}_p\end{array}\right).\label{a22inv}
\end{equation}
Now, by (\ref{IPhi}) and (\ref{a22inv}), we have
\begin{align*}
&{\bf A}_{12}(q){\bf A}_{22}(q)^{-1}{\bf A}_{21}(q)\\
&=\left(\begin{array}{cc}\bphi^-_{q0}\otimes\Phi_0&\phi_{p0}\Phi_0
\end{array}\right){\bf A}_{22}(q)^{-1}\left(\begin{array}{c}{\boldsymbol\delta}_{q-1}\otimes\Phi_0\\{\bf 0}_{p\times p}\end{array}
\right)\\
&=\left\{(\bphi^-_{q0}\otimes\Phi_0)+(\phi_{p0}{\bf K}_{q-1}'\otimes\Phi_0^2)\right\}{\bf A}_{22}(q-1)^{-1}({\boldsymbol\delta}_{q-1}\otimes\Phi_0)\\
&=\left(\begin{array}{cc}\bphi^-_{q-1,0}\otimes\Phi_0&\phi_{q0}\Phi_0+\phi_{p0}\Phi_0^2\end{array}\right){\bf A}_{22}(q-1)^{-1}
\left(\begin{array}{c}{\boldsymbol\delta}_{q-2}\otimes\Phi_0\\{\bf 0}_{p\times p}\end{array}
\right)\\
&=\left\{(\bphi^-_{q-1,0}\otimes\Phi_0)+(\phi_{q0}{\bf K}_{q-2}'\otimes\Phi_0^2)
+(\phi_{p0}{\bf K}_{q-2}'\otimes\Phi_0^3)\right\}{\bf A}_{22}(q-2)^{-1}
({\boldsymbol\delta}_{q-2}\otimes\Phi_0)\\
&=\left(\begin{array}{cc}\bphi^-_{q-2,0}\otimes\Phi_0&\phi_{q-1,0}\Phi_0+\phi_{q0}\Phi_0^2+\phi_{p0}\Phi_0^3\end{array}\right)
{\bf A}_{22}(q-2)^{-1}
\left(\begin{array}{c}{\boldsymbol\delta}_{q-3}\otimes\Phi_0\\{\bf 0}_{p\times p}\end{array}\right)\\
&=...\\
&=\left(\begin{array}{cc}\phi_{20}\Phi_0&\sum_{k=3}^p\phi_{k0}\Phi_0^{k-2}\end{array}\right){\bf A}_{22}(2)^{-1}
\left(\begin{array}{c}\Phi_0\\{\bf 0}_{p\times p}\end{array}\right).
\end{align*}
Thus, observing that
\begin{equation}
{\bf A}_{22}(2)={\bf I}_{2p}-{\bf J}_2\otimes\Phi_0=
\left(\begin{array}{cc}{\bf I}_p&{\bf 0}_{p\times p}\\
-\Phi_0&{\bf I}_p\end{array}\right),\label{a222inv}
\end{equation}
which gives
$${\bf A}_{22}(2)^{-1}=\left(\begin{array}{cc}{\bf I}_p&{\bf 0}_{p\times p}\\\Phi_0&{\bf I}_p\end{array}\right),$$
we get
\begin{align*}
{\bf A}_{12}(q){\bf A}_{22}(q)^{-1}{\bf A}_{21}(q)
&=\left(\begin{array}{cc}\phi_{20}\Phi_0&\sum_{k=3}^p\phi_{k0}\Phi_0^{k-2}\end{array}\right)\left(\begin{array}{cc}{\bf I}_p&{\bf 0}_{p\times p}\\\Phi_0&{\bf I}_p\end{array}\right)
\left(\begin{array}{c}\Phi_0\\{\bf 0}_{p\times p}\end{array}\right)\\
&=\sum_{k=2}^p\phi_{k0}\Phi_0^k,
\end{align*}
and (\ref{IPhi})-(\ref{a11inv}) imply
\begin{equation}
{\bf A}^{11}=\left({\bf I}_p-\sum_{k=1}^p\phi_{k0}\Phi_0^k\right)^{-1}.\label{a11inv2}\end{equation}

Next, again from Mardia et al (1979), p.459, we have
\begin{equation}
{\bf A}^{21}(q)=-{\bf A}_{22}(q)^{-1}{\bf A}_{21}(q){\bf A}^{11},
\label{a21}
\end{equation}
and defining
$${\bf L}_{q-1}=\left(\begin{array}{c}
{\bf I}_{p(q-1)}\\{\bf K}_{q-1}'\otimes\Phi_0\end{array}\right),$$
we find from (\ref{IPhi}) and (\ref{a22inv}) that
\begin{align}
&-{\bf A}_{22}(q)^{-1}{\bf A}_{21}(q)\notag\\
&=\left(\begin{array}{cc}
{\bf A}_{22}(q-1)^{-1}&{\bf 0}_{p(q-1)\times p}\\
({\bf K}_{q-1}'\otimes\Phi_0){\bf A}_{22}(q-1)^{-1}&{\bf I}_p\end{array}\right)\left(\begin{array}{c}{\boldsymbol\delta}_{q-1}\otimes\Phi_0\\
{\bf 0}_{p\times p}\end{array}\right)\notag\\
&={\bf L}_{q-1}{\bf A}_{22}(q-1)^{-1}({\boldsymbol\delta}_{q-1}\otimes\Phi_0)\notag\\
&={\bf L}_{q-1}\left(\begin{array}{cc}
{\bf A}_{22}(q-2)^{-1}&{\bf 0}_{p(q-2)\times p}\\
({\bf K}_{q-2}'\otimes\Phi_0){\bf A}_{22}(q-2)^{-1}&{\bf I}_p\end{array}\right)\left(\begin{array}{c}{\boldsymbol\delta}_{q-2}\otimes\Phi_0\\
{\bf 0}_{p\times p}\end{array}\right)\notag\\
&={\bf L}_{q-1}{\bf L}_{q-2}{\bf A}_{22}(q-2)^{-1}({\boldsymbol\delta}_{q-2}\otimes\Phi_0)=...\notag\\
&={\bf L}_{q-1}{\bf L}_{q-2}\cdots{\bf L}_2{\bf A}_{22}(2)^{-1}({\boldsymbol\delta}_{2}\otimes\Phi_0).\label{a22a21}
\end{align}
Here, observe that
\begin{align*}
&{\bf L}_{q-1}{\bf L}_{q-2}\\
&=\left(\begin{array}{c}
{\bf I}_{p(q-1)}\\{\bf K}_{q-1}'\otimes\Phi_0\end{array}\right)
\left(\begin{array}{c}
{\bf I}_{p(q-2)}\\{\bf K}_{q-2}'\otimes\Phi_0\end{array}\right)\\
&=\left(\begin{array}{cc}
{\bf I}_{p(q-2)}&{\bf 0}_{p(q-2)\times p}\\
{\bf 0}_{p\times p(q-2)}&{\bf I}_p\\
{\bf 0}_{p\times p(q-2)}&\Phi_0\end{array}\right)
\left(\begin{array}{c}
{\bf I}_{p(q-2)}\\{\bf K}_{q-2}'\otimes\Phi_0\end{array}\right)
=\left(\begin{array}{c}
{\bf I}_{p(q-2)}\\{\bf K}_{q-2}'\otimes\Phi_0\\
{\bf K}_{q-2}'\otimes\Phi_0^2\end{array}\right),
\end{align*}
and similarly,
\begin{align*}
&{\bf L}_{q-1}{\bf L}_{q-2}{\bf L}_{q-3}\\
&=\left(\begin{array}{c}
{\bf I}_{p(q-2)}\\{\bf K}_{q-2}'\otimes\Phi_0\\
{\bf K}_{q-2}'\otimes\Phi_0^2\end{array}\right)
\left(\begin{array}{c}
{\bf I}_{p(q-3)}\\{\bf K}_{q-3}'\otimes\Phi_0\end{array}\right)\\
&=\left(\begin{array}{cc}
{\bf I}_{p(q-3)}&{\bf 0}_{p(q-3)\times p}\\
{\bf 0}_{p\times p(q-3)}&{\bf I}_p\\
{\bf 0}_{p\times p(q-3)}&\Phi_0\\
{\bf 0}_{p\times p(q-3)}&\Phi_0^2\end{array}\right)
\left(\begin{array}{c}
{\bf I}_{p(q-3)}\\{\bf K}_{q-3}'\otimes\Phi_0\end{array}\right)
=\left(\begin{array}{c}
{\bf I}_{p(q-3)}\\{\bf K}_{q-3}'\otimes\Phi_0\\
{\bf K}_{q-3}'\otimes\Phi_0^2\\{\bf K}_{q-3}'\otimes\Phi_0^3
\end{array}\right),
\end{align*}
so that by recursion,
$${\bf L}_{q-1}{\bf L}_{q-2}\cdots{\bf L}_2
=\left(\begin{array}{c}
{\bf I}_{2p}\\{\bf K}_2'\otimes\Phi_0\\
{\bf K}_2'\otimes\Phi_0^2\\
\vdots\\{\bf K}_2'\otimes\Phi_0^{q-2}
\end{array}\right).$$
Hence, because from (\ref{a222inv}),
$${\bf A}_{22}(2)^{-1}({\boldsymbol\delta}\otimes\Phi_0)
=\left(\begin{array}{cc}{\bf I}_p&{\bf 0}_{p\times p}\\
\Phi_0&{\bf I}_p\end{array}\right)
\left(\begin{array}{c}\Phi_0\\{\bf 0}_{p\times p}\end{array}\right)
=\left(\begin{array}{c}\Phi_0\\\Phi_0^2\end{array}\right),$$
we get via (\ref{a22a21}) that
$$-{\bf A}_{22}(q)^{-1}{\bf A}_{21}(q)=
\left(\begin{array}{c}
{\bf I}_{2p}\\{\bf K}_2'\otimes\Phi_0\\{\bf K}_2'\otimes\Phi_0^2\\
\vdots\\{\bf K}_2'\otimes\Phi_0^{q-2}\end{array}\right)
\left(\begin{array}{c}\Phi_0\\\Phi_0^2\end{array}\right)
=\left(\begin{array}{c}
\Phi_0\\\Phi_0^2\\\vdots\\\Phi_0^q\end{array}\right),$$
i.e. by (\ref{a11inv2}) and (\ref{a21}),
\begin{equation}
\left(\begin{array}{c}{\bf A}^{11}\\{\bf A}^{21}(q)\end{array}\right)
=\left(\begin{array}{c}
{\bf I}_p\\\Phi_0\\\vdots\\\Phi_0^q\end{array}\right){\bf A}^{11},\quad
{\bf A}^{11}=\left({\bf I}_p-\sum_{k=1}^p\phi_{k0}\Phi_0^k\right)^{-1}.\label{a11a21}
\end{equation}
Moreover, since ${\bf M}={\boldsymbol\delta}_p{\boldsymbol\delta}_p'$, (\ref{vecid}) implies
\begin{equation}
\ve({\bf M})={\boldsymbol\delta}_p\otimes{\boldsymbol\delta}_p
=\left(\begin{array}{c}{\boldsymbol\delta}_p\\{\bf 0}_{p(p-1)}\end{array}\right),\label{vecM}
\end{equation}
which together with (\ref{IPhiinv}) yields
\begin{align}
&\ve({\bf M})'({\bf I}_{p^2}-\bPhi)^{-1}\ve({\bf M})\notag\\
&=\left(\begin{array}{cc}{\boldsymbol\delta}_p'&{\bf 0}_{p(p-1)}'
\end{array}\right)
\left(\begin{array}{cc}{\bf A}^{11}&{\bf A}^{12}(q)\\{\bf A}^{21}(q)&{\bf A}^{22}(q)\end{array}\right)\left(\begin{array}{c}{\boldsymbol\delta}_p\\{\bf 0}_{p(p-1)}\end{array}\right)
={\boldsymbol\delta}_p'{\bf A}^{11}{\boldsymbol\delta}_p.
\label{MIPhiM}
\end{align}
Similarly, (\ref{IPhiinv}), (\ref{a11a21}) and (\ref{vecM}) give
\begin{align*}
&({\bf I}_{p^2}-\bPhi)^{-1}\ve({\bf M})\\
&=\left(\begin{array}{cc}{\bf A}^{11}&{\bf A}^{12}(q)\\{\bf A}^{21}(q)&{\bf A}^{22}(q)\end{array}\right)\left(\begin{array}{c}{\boldsymbol\delta}_p\\{\bf 0}_{p(p-1)}\end{array}\right)
=\left(\begin{array}{c}
{\bf I}_p\\\Phi_0\\\vdots\\\Phi_0^q\end{array}\right){\bf A}^{11}
{\boldsymbol\delta}_p,
\end{align*}
so that
\begin{align}
&\ve^{-1}\left\{({\bf I}_{p^2}-\bPhi)^{-1}\ve({\bf M})\right\}\notag\\
&=\left(\begin{array}{cccc}{\bf A}^{11}{\boldsymbol\delta}_p&
\Phi_0{\bf A}^{11}{\boldsymbol\delta}_p&\ldots&\Phi_0^q{\bf A}^{11}{\boldsymbol\delta}_p\end{array}\right).\label{veinv}
\end{align}

Next, we need to get a suitable expression for $\Omega_0^{-1}$, which is the inverse of the asymptotic covariance matrix of $\hat\bphi-\bphi$. At first, note that we can write
\begin{align}
\hat\bphi-\bphi_0&=\left(\sum_{t=1}^n\ym\ym'\right)^{-1}
\sum_{t=1}^n\ym(y_t-\ym'\bphi_0)\notag\\
&=\left(\sum_{t=1}^n\ym\ym'\right)^{-1}
\sum_{t=1}^n\ym\varepsilon_t.\label{phihatphi}
\end{align}
Now, from (\ref{Evsumy2}) and the law of large numbers, we find
\begin{equation}
n^{-1}\sum_{t=1}^n\ym\ym'=E\left(\sum_{t=1}^n\ym\ym'\right)+O_p(n^{-1})=\sigma^2{\bf R}+O_p(n^{-1}),\label{sumylim}
\end{equation}
where
$${\bf R}=\ve^{-1}\left\{({\bf I}_{p^2}-\bPhi)^{-1}\ve({\bf M})\right\}.$$
Hence, by the Slutsky theorem, to find the limiting covariance matrix of (\ref{phihatphi}), it is enough to find the limiting covariance matrix of 
$\sum_{t=1}^n\ym\varepsilon_t$. Because its expectation is zero, this matrix is given from the expectation of (using (\ref{ym}))
\begin{align*}
&\left(\sum_{s=1}^n{\bf y}_{s-}\varepsilon_s\right)\left(\sum_{t=1}^n\ym\varepsilon_t\right)'\\
&=\sum_{s=1}^n\sum_{t=1}^n{\bf y}_{s-}\varepsilon_s\varepsilon_t\ym'
=\sum_{s=1}^n\sum_{t=1}^n
\sum_{i=1}^{s-1}\varepsilon_s\varepsilon_t\Phi_0^{s-1-i}\te_i
\left(\sum_{j=1}^{t-1}\Phi_0^{t-1-j}\te_j\right)'\\
&=\sum_{s=1}^n\sum_{t=1}^n
\sum_{i=1}^{s-1}\sum_{j=1}^{t-1}\varepsilon_s\varepsilon_t\Phi_0^{s-1-i}\te_i\te_j'(\Phi_0')^{t-1-j},
\end{align*}
which is, using $E(\varepsilon_s\varepsilon_t)=\sigma^2$ if $s=t$ and 0 otherwise and $E(\te_i\te_j')=\sigma^2{\bf M}$ if $i=j$ and 0 otherwise and proceeding as in (\ref{sumPhiMPhi}),
\begin{align*}
&E\left\{\left(\sum_{s=1}^n{\bf y}_{s-}\varepsilon_s\right)\left(\sum_{t=1}^n\ym\varepsilon_t\right)'\right\}\\
&=\sigma^4\sum_{s=1}^n
\sum_{i=1}^{s-1}\Phi_0^{s-1-i}{\bf M}(\Phi_0')^{t-1-j}
=n\sigma^4{\bf R}+O(1).
\end{align*} 
Thus, (\ref{phihatphi}), (\ref{sumylim}) and the Slutsky theorem imply that $n$ times the covariance matrix of $\hat\bphi-\bphi_0$ is
\begin{align*}
\Omega_0&=E\left[\left(n^{-1}\sum_{t=1}^n\ym\ym'\right)^{-1}
n^{-1}\left\{\left(\sum_{s=1}^n{\bf y}_{s-}\varepsilon_s\right)\left(\sum_{t=1}^n\ym\varepsilon_t\right)'\right\}\right.\\
&\cdot\left.\left(n^{-1}\sum_{t=1}^n\ym\ym'\right)^{-1}\right]\\
&={\bf R}^{-1}{\bf R}{\bf R}^{-1}+O(n^{-1})
={\bf R}^{-1}+O(n^{-1}),
\end{align*}
and (\ref{Bn}) and (\ref{MIPhiM}) yield
\begin{align}
\frac{1}{n\sigma^2}B&=(\bphi_0-\bphi)'{\bf R}(\bphi_0-\bphi)-\ve({\bf M})'({\bf I}_{p^2}-\bPhi)^{-1}\ve({\bf M})\notag\\
&+\bphi'{\bf R}(2\bphi_0-\bphi)+O_p(n^{-1/2})\notag\\
&=\bphi_0'{\bf R}\bphi_0-{\boldsymbol\delta}_p'{\bf A}^{11}{\boldsymbol\delta}_p+O_p(n^{-1/2}).\label{Bn2}
\end{align}
But now, by (\ref{veinv}) and since $\bphi_0'={\boldsymbol\delta}_p'\Phi_0$,
\begin{align*}
\bphi_0'{\bf R}\bphi_0&=\bphi_0'\left(\begin{array}{cccc}{\bf A}^{11}{\boldsymbol\delta}_p&
\Phi_0{\bf A}^{11}{\boldsymbol\delta}_p&\ldots&\Phi_0^q{\bf A}^{11}{\boldsymbol\delta}_p\end{array}\right)\bphi_0\\
&={\boldsymbol\delta}_p'\left(\begin{array}{cccc}\Phi_0{\bf A}^{11}{\boldsymbol\delta}_p&
\Phi_0^2{\bf A}^{11}{\boldsymbol\delta}_p&\ldots&\Phi_0^p{\bf A}^{11}{\boldsymbol\delta}_p\end{array}\right)\left(\begin{array}{c}
\phi_{10}\\\phi_{20}\\\vdots\\\phi_{p0}\end{array}\right)\\
&={\boldsymbol\delta}_p'\left(\sum_{k=1}^p\phi_{k0}\Phi_0^k\right)
{\bf A}^{11}{\boldsymbol\delta}_p,
\end{align*}
and using (\ref{a11inv2}), the main term of (\ref{Bn2}) becomes
$$\bphi_0'{\bf R}\bphi_0-{\boldsymbol\delta}_p'{\bf A}^{11}{\boldsymbol\delta}_p={\boldsymbol\delta}_p'\left(\sum_{k=1}^p\phi_{k0}\Phi_0^k-{\bf I}_p\right)
{\bf A}^{11}{\boldsymbol\delta}_p
=-{\boldsymbol\delta}_p'{\boldsymbol\delta}_p=-1.$$
Conclusively, (\ref{Bn2}) gives
$$\frac{1}{n\sigma^2}B
=-1+O_p(n^{-1/2}).$$
Inserting into (\ref{frac}) gives the result.

\section*{References}

Brockwell, P.J., and Davies R.A. (2002) \emph{Introduction to Time Series and Forecasting, 2nd ed.}, New York: Springer.
\\[1.4mm]
Cox, D. R. (1958) "Some problems connected with statistical inference," \emph{Annals of Mathematical Statistics}, 29, 357-372.
\\[1.4mm]
Fisher, R. A. (1930) "Inverse probability," \emph{Proceedings of the Cambridge Philosophical Society}, 26, 528-535.
\\[1.4mm]
Hald, A. (2007) \emph{A History of Parametric Statistical Inference from Bernoulli to Fisher, 1713-1935}, New York: Springer.
\\[1.4mm]
Gnedenko, B. V. (1989) \emph{The Theory of Probability, 4th ed.}, New York: Chelsea Publishing Company.
\\[1.4mm]
Jeffreys, H. (1931) \emph{Theory of Probability}, Cambridge: Cambridge University Press.
\\[1.4mm]
Knowles, I., and Renka, R. J.  (2014) "Methods for numerical differentiation of noisy data," \emph{Electronic Journal of Differential Equations, Conference 21}, 235-246.
\\[1.4mm]
Lancaster, P., and Farahat, H. K. (1972) "Norms on direct sums and tensor products," \emph{Mathematics of Computation}, 118, 401-414.
\\[1.4mm]
Larsson, R (2024) "Confidence distributions for the autoregressive parameter," \emph{The American Statistician},
78, 58-65.
\\[1.4mm]
Mann, H. B., and Wald, A.  (1943) "On stochastic limit and order relationships," \emph{Annals of Mathematical Statistics}, 14, 217-226.
\\[1.4mm]
Mardia, K.V., Kent, J.T. and Bibby, J.M. (1979) \emph{Multivariate Analysis}, London: Academic Press.
\\[1.4mm]
Neyman, J. (1941) "Fiducial argument and the theory of confidence intervals," \emph{Biometrika}, 32, 128-150.
\\[1.4mm]
Pawitan, Y. and Lee, Y (2021) "Confidence as likelihood." \emph{Statistical Science} 36, 509-517.
\\[1.4mm]
Schweder, T., and Hjort, N. L. (2016) \emph{Confidence, Likelihood, Probability - Statistical Inference with Confidence Distributions}, Cambridge University Press: New York.
\\[1.4mm]
Shumway, R. H., Stoffer, D. S. (2017) \emph{Time Series Analysis and its Applications -- With R Examples. 4th edition},
New York: Springer.
\\[1.4mm]
Wei, W.W.S. (2006) \emph{Time Series Analysis -- Univariate and Multivariate Methods. 2nd edition}, Pearson Education.
\\[1.4mm]
Xie, M., and Singh, K. (2013) "Confidence distribution, the frequentist distribution estimator of a parameter: a review," \emph{International Statistical review}, 81, 3-39.
\\[1.4mm]

\end{document}